\documentclass[letterpaper, 10 pt, conference]{ieeeconf}
\IEEEoverridecommandlockouts

\usepackage{graphicx}
\usepackage{amsmath,amssymb,bm}
\let\labelindent\relax
\usepackage{enumitem}
\usepackage{url}
\usepackage{xspace}
\usepackage{booktabs}
\usepackage{colortbl}

\newcommand{\real}{\ensuremath{\mathbb{R}}}
\newcommand{\nat}{\ensuremath{\mathbb{N}}}
\newcommand{\compl}{\ensuremath{\mathbb{C}}}
\newcommand{\smat}[1]{\ensuremath{\left[\begin{smallmatrix}#1\end{smallmatrix}\right]}}
\newcommand{\bmat}[1]{\ensuremath{\begin{bmatrix}#1\end{bmatrix}}}

\newcommand{\tu}[1]{{\textup{#1}}}
\newcommand{\mb}[1]{\mathbf{#1}}
\newcommand{\sa}[1]{\mathsf{#1}}
\newcommand{\cnri}{\ensuremath{z = x + j y \in \compl \colon}}
\newcommand{\ri}{\ensuremath{x + j y \in \compl \colon}}
\newcommand{\pseudoX}{\textsl{\textsf{x}}}
\newcommand{\co}{\textup{c}}
\newcommand{\si}{\textup{s}}

\DeclareMathOperator{\He}{He}
\DeclareMathOperator{\Tr}{Tr}
\DeclareMathSymbol{\cdoT}{\mathord}{symbols}{"01}

\newtheorem{remark}{Remark}
\newtheorem{lemma}{Lemma}
\newtheorem{corollary}{Corollary}
\newtheorem{theorem}{Theorem}
\newtheorem{proposition}{Proposition}
\newtheorem{assumption}{Assumption}
\newtheorem{definition}{Definition}
\newtheorem{fact}{Fact}
\newtheorem{example}{Example}

\newcommand{\grayRow}{\rowcolor[rgb]{0.973,0.973,0.973}}

\allowdisplaybreaks

\title{\LARGE \bf Learning controllers for performance through LMI regions}
\author{Andrea Bisoffi$^{1}$, Claudio De Persis$^{1}$ and Pietro Tesi$^{2}$
\thanks{$^\star$This research is partially supported by a Marie Sk{\l}odowska-Curie COFUND grant, no.~754315 and by NWO, project no.~15472.}
\thanks{$^{1}$A. Bisoffi and C. De Persis are with ENTEG and J.C. Willems Center for Systems and Control, University of Groningen, 9747 AG Groningen, The Netherlands
{\tt\small \{a.bisoffi,c.de.persis\}@rug.nl}}%
\thanks{$^{2}$P. Tesi is with DINFO, University of Florence, 50139 Florence, Italy         {\tt\small \{pietro.tesi\}@unifi.it}}%
}

\begin{document}

\maketitle
\thispagestyle{empty}
\pagestyle{empty}

\begin{abstract}
In an open-loop experiment, an input sequence is applied to an unknown linear time-invariant system (in continuous or discrete time) affected also by an unknown-but-bounded disturbance sequence (with an energy or instantaneous bound); the corresponding state sequence is measured.
The goal is to design directly from the input and state sequences a controller that enforces a certain performance specification on the transient behaviour of the unknown system.
The performance specification is expressed through a subset of the complex plane where closed-loop eigenvalues need to belong, a so called LMI region.
For this control design problem, we provide here convex programs to enforce the performance specification from data in the form of linear matrix inequalities (LMI).
For generic LMI regions, these are sufficient conditions to assign the eigenvalues within the LMI region for all possible dynamics consistent with data, and become necessary and sufficient conditions for special LMI regions.
In this way, we extend classical model-based conditions from a seminal work in the literature to the setting of data-driven control from noisy data.
Through two numerical examples, we investigate how these data-based conditions compare with each other.
\end{abstract}

\section{Introduction}

Whenever it is challenging or cumbersome to derive a model for a process to be controlled or to identify unambiguously its parameters, a viable alternative is to bypass these two steps altogether and, from data collected on the process, design directly a (feedback) controller \cite{sznaier2021survey}.
Direct data-driven control was conceived within the discipline of system identification, and is enjoying renewed popularity thanks a fundamental result by Willems et al. \cite[Thm.~1]{willems2005note} for linear systems and noiseless data, see also \cite{coulson2018deepc,depersis2020tac}. 
A natural continuation within linear systems has been how to handle the realistic case of noisy data, whose induced uncertainty has been addressed via tools from robust control.
With bounded noise and noisy input-state data points collected in an open-loop experiment, one ends up with a \emph{set} of dynamical matrices $(A,B)$ consistent with data and wants to design a controller that guarantees certain properties of the closed-loop system for all such $(A,B)$.
(Necessary and) sufficient conditions (typically in the form of convenient convex programs) were given in the cases of stabilization \cite{depersis2020tac,dai2020ifac}, linear quadratic regulation \cite{mania2019certainty,depersis2021lowcomplexity,xue2020datadriven}, and dynamic performance \cite{berberich2019robust,berberich2020combining,vanwaarde2020noisy}.
Imposing a certain performance specification for the process to be controlled is, in applications, as relevant as stabilization; however, for these data-based control designs, dynamic performance has been less investigated than stabilization and, to the best of our knowledge, only in terms of quadratic \cite{berberich2019robust,berberich2020combining}, $\mathcal{H}_2$ \cite{vanwaarde2020noisy,berberich2020combining} or $\mathcal{H}_\infty$ performance \cite{vanwaarde2020noisy}.

An alternative method to impose performance specifications is by imposing  that the closed-loop eigenvalues belong to specific subsets of the complex plane.
Indeed, some salient characteristics of the closed-loop transient response (in continuous time) depend on these subsets of the complex plane: e.g., convergence rate is greater than $\ell>0$ if all eigenvalues have real part less than $-\ell$, damping ratio is greater than $\cos \theta$ if eigenvalues are within a cone with vertex in $0$ and aperture $2 \theta$, and eigenvalues in the intersection of these two sets with suitable $\ell$ and $\theta$ achieve a fast response with limited overshoot, see Example~\ref{example:wedge} later.
The fundamental work in~\cite{chilali1996hInf} showed that for a certain subset of the complex plane and a given model $(A,B)$, finding the feedback gain $K$ by which all eigenvalues of $A + B K$ belong to that subset is equivalent to solving a linear matrix inequality (LMI), see \cite[Thm.~2.2]{chilali1996hInf} recalled later in Fact~\ref{fact:S-stab}; such subsets take thus the name of LMI regions.
Notably, the open left halfplane and the open unit disk are LMI regions, and a large number of subsets of the complex plane can be expressed as LMI regions, see Fig.~\ref{fig:basic LMI regions}; moreover, the intersection of LMI regions is also \emph{equivalently} associated with the conjunction of the respective LMIs, see \cite[Cor.~2.3]{chilali1996hInf} recalled later in Fact~\ref{fact:inters}, to the effect that LMI regions are dense in the set of convex regions that are  symmetric with respect to the real axis \cite[\S II.C]{chilali1996hInf}.
In the approach for performance by LMI regions, appealing features are then that they equivalently give rise to convex inequalities and can express, or at least approximate closely, the subsets of the complex plane relevant for control purposes.
Finally, the approach with LMI regions does not exclude using also $\mathcal{H}_2$ and $\mathcal{H}_\infty$ approaches \cite[\S III]{chilali1996hInf}; however, with respect to them, it implements in an easy way performance specifications by linking the desired characteristics of the time response to regions of the complex plane that are well known to a control engineer familiar with frequency methods and loop shaping. 
Indeed, the approach by LMI regions has been used effectively in experimental applications \cite{olalla2011optimal,pereira2013multiple,poussotvassal2016gust,cocetti2020hybrid}.

All these positive features in the model-based case appear promising and have motivated us to study how to impose performance specification through LMI regions also in the data-based case.
In this case we need to assign the eigenvalues in LMI regions for all matrices $(A,B)$ consistent with data.
From a conceptual viewpoint, this is similar in nature to \cite{chilali1999robust} that investigates robustness of pole clustering in LMI regions with respect to complex unstructured and real structured uncertainty (such as parameter uncertainty); here we consider robustness with respect to a different type of uncertainty, namely, that induced by noisy data.
Our contribution is that we provide sufficient conditions to design a controller enforcing robust eigenvalue assignment in spite of noisy data for generic LMI regions and their intersections, under noise models with an energy bound on the whole noise sequence of the experiment and with an instantaneous bound on each noise element of the sequence; moreover, we obtain that these sufficient conditions become also necessary for special LMI regions; finally, all these results hold both for continuous and discrete time and are given in terms of convenient linear matrix inequalities.
The proposed data-driven approach based on LMI regions, which has not been investigated so far, constitutes an alternative method to $\mathcal{H}_2$ and $\mathcal{H}_\infty$ approaches to guarantee performance.
In a nutshell, the approach features an experiment for data collection, performance specifications are intuitively expressed as LMI regions, and the proposed convex programs design the controller to enforce the specification automatically; we then believe that the approach has the potential to incentivize these new data-based techniques among control engineers.

\textbf{\itshape Structure:} In Section~\ref{sec:review}, we report the notions needed from the model-based setting in~\cite{chilali1996hInf}.
In Section~\ref{sec:setting} we formulate the data-based problem.
In Section~\ref{sec:suff cond}, we give sufficient conditions for generic LMI regions, whereas, for special LMI regions, necessary and sufficient conditions are given in Section~\ref{sec:nec and suff cond}. Other relevant conditions for generic LMI regions are in Section~\ref{sec:suff cond alt}.
How all these conditions compare is investigated numerically in Section~\ref{sec:num ex}.

\textbf{\itshape Notation}: 
$\nat_{\ge 1}$ denotes the natural numbers $1$, $2$, \dots; $\real$ denotes the real numbers; $\compl$ denotes the complex numbers.
For $z \in \compl$, $\bar z$ denotes the complex conjugate of $z$.
Given $n \in \nat_{\ge 1}$, $I_n$ (or $I$) denotes an identity matrix of dimension $n$ (or of suitable dimension).
For a matrix with complex entries, the hermitian operator is $\He A := A + A^{\sf{H}}$; for a matrix with real entries, the transposition operator is $\Tr A := A + A^\top$.
For symmetric matrices $A$ and $C$, we sometimes abbreviate a symmetric matrix $\smat{A & B^\top\\ B & C}$ as $\smat{A & \star \\ B & C}$ or $\smat{A & B^\top\\ \star & C}$.
Positive definiteness (semidefiniteness, respectively) of a symmetric matrix $A$ is indicated as $A \succ 0$ ($A \succeq 0$, respectively).
For a $A=A^\top \succeq 0$, $A^{1/2}$ denotes the unique positive semidefinite root of $A$.
The Kronecker product is denoted by $\otimes$ and the standard properties of the Kronecker product we use can be found in \cite[\S 4.2]{horn1994topics}.

\section{Review on LMI regions}
\label{sec:review}

This section recalls the notions we need from \cite{chilali1996hInf}.

\begin{definition}{\cite[Def.~2.1]{chilali1996hInf}} A subset $\mathcal{S}$ of the complex plane is called an LMI region if for some $s \in \nat_{\ge 1}$, there exists a symmetric matrix $\alpha \in \real^{s\times s}$ and a matrix $\beta \in \real^{s \times s}$ such that
\begin{equation}
\label{set S generic}
\mathcal{S} = \{ z \in \compl \colon \alpha + z \beta + \bar z \beta^\top \prec 0 \}
\end{equation}
where the matrix $\alpha + z \beta + \bar z \beta^\top$ is Hermitian. $(\alpha,\beta)$ are called data of $\mathcal{S}$.
\end{definition}
A generic $s \in \nat_{\ge 1}$ is possible; however, $s=2$ is a convenient trade-off between tractability and expressivity since it allows expressing a plethora of well-known quadratic curves.
Most common ones are shown in Fig.~\ref{fig:basic LMI regions}.
By expressing $z$ in~\eqref{set S generic} in terms of its real and imaginary parts, we write \eqref{set S generic} for $s=2$ as
\begingroup
\setlength{\arraycolsep}{1.5pt}
\medmuskip=1mu plus 2mu
\thickmuskip=1.5mu plus 3mu
\begin{align}
& \hspace*{-5pt}\mathcal{S} = \{ \ri   \notag \\
& \hspace*{15pt} \smat{\alpha_{11} & \alpha_{12}\\ \alpha_{12} & \alpha_{22}} + (x + j y) \smat{\beta_{11} & \beta_{12} \\ \beta_{21} & \beta_{22} }  +  (x - j y) \smat{\beta_{11} & \beta_{21} \\ \beta_{12} & \beta_{22}}  \prec 0 \}  \notag \\
& \hspace*{1pt}= \{ \ri \alpha_{11} + 2 x \beta_{11} < 0,  \notag \\
& \hspace*{17pt} \alpha_{11} \alpha_{22} - \alpha_{12}^2 + 2x \big(\alpha_{11}\beta_{22} + \alpha_{22}\beta_{11} - \alpha_{12}(\beta_{12} + \beta_{21}) \big)\notag \\
& \hspace*{17pt} - x^2 \big( (\beta_{12}+\beta_{21})^2 - 4 \beta_{11}\beta_{22} \big)
- y^2 (\beta_{12} - \beta_{21})^2 > 0\}. \!\! \label{generic LMI region s=2}
\end{align}
\endgroup
The $\alpha$ and $\beta$ corresponding to the regions in Fig.~\ref{fig:basic LMI regions} are reported in the appendix, in Table~\ref{tab:sets}.

\begin{figure}
\centerline{\hspace*{8pt}\includegraphics[scale=.47]{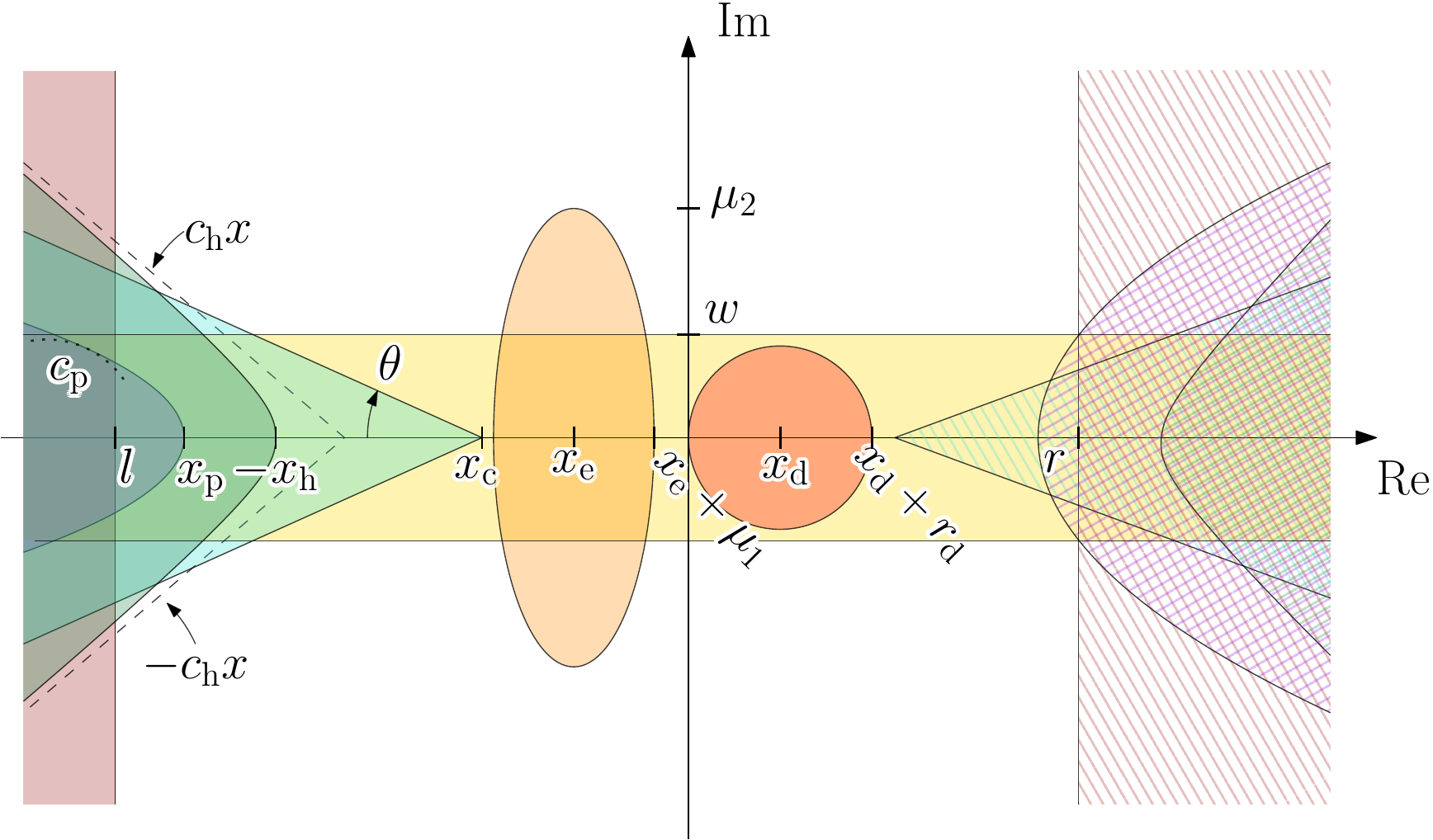}}%
\caption{Regions of the complex plane that can be represented as LMI regions for $s=2$: facing left (filled) and facing right (hatched).
The only constrained parameters of these regions are: $r_{\tu{d}} > 0$ (radius of disk), $w> 0$ (semiwidth of horizontal strip), $\mu_1 > 0$ and $\mu_2 > 0$ (semiaxes of ellipsoid), $c_{\tu{p}} > 0$ (curvature of parabola), $x_{\tu{h}} > 0$ and $c_{\tu{h}} > 0$  (vertex and angular coefficient of asymptotes of hyperbola), $\theta \in (0,\pi/2)$ (semiaperture of cone). All other parameters ($l$, $r$, $x_{\tu{d}}$, $x_{\tu{e}}$, $x_{\tu{p}}$, $x_{\tu{c}}$) are free.}
\label{fig:basic LMI regions}
\end{figure}

The next definition expresses in short that eigenvalues of a matrix belong to a certain LMI region.
\begin{definition}{\cite[p.~359]{chilali1996hInf}} For $\mathcal{S} \subseteq \compl$, the matrix $A$ is $\mathcal{S}$-stable if all eigenvalues of $A$ lie in $\mathcal{S}$.
\end{definition}

We recall a first elegant result from \cite{chilali1996hInf}.
\begin{fact}{\cite[Thm.~2.2]{chilali1996hInf}}
\label{fact:S-stab}
For an LMI region $\mathcal{S} \subseteq \compl$ with data $(\alpha,\beta)$, the matrix $A$ is $\mathcal{S}$-stable if and only if there exists a symmetric matrix $P$ such that
\begin{equation}
\label{S-stab cond}
P \succ 0, \alpha \otimes P + \beta \otimes (AP) + \beta^\top \otimes (PA^\top) \prec 0.
\end{equation}
\end{fact}
\smallskip

As it emerges from the proof of \cite[Thm.~2.2]{chilali1996hInf}, an LMI region $\mathcal{S}$ is not limited to be in the left halfplane.
Then, Fact~\ref{fact:S-stab} enables treating continuous and discrete time simultaneously.
For an LMI region $\mathcal{S}$ with data $(\alpha, \beta)$, define its characteristic matrix $M_\mathcal{S}$ as
\begin{equation}
\label{charact matrix}
M_\mathcal{S}(A,P) := \alpha \otimes P + \beta \otimes (AP) + \beta^\top \otimes (PA^\top),
\end{equation}
so that $M_\mathcal{S}(A,P) \prec 0$ is precisely the main condition in \eqref{S-stab cond}.
This brings us to a second key result from \cite{chilali1996hInf}, recalled next.
\begin{fact}{\cite[Cor.~2.3]{chilali1996hInf}}
\label{fact:inters}
Given two LMI regions $\mathcal{S}_1$ with data $(\alpha_1, \beta_1)$ and $\mathcal{S}_2$ with data $(\alpha_2,\beta_2)$, a matrix $A$ is both $\mathcal{S}_1$-stable and $\mathcal{S}_2$-stable if and only if there exists a symmetric positive definite matrix $P$ such that $M_{\mathcal{S}_1}(A,P) \prec 0$ and $M_{\mathcal{S}_2}(A,P) \prec 0$.
\end{fact}

From linear systems theory, Hurwitz and Schur stability of a matrix $A$ correspond to eigenvalues of $A$ lying respectively in the open left halfplane
\begin{equation}
\label{Hurwitz as LMI region}
\mathcal{S}_{\tu{H}} := \{ \ri x < 0 \} = \{ z \in \compl \colon 0 + z + \bar z < 0\}
\end{equation}
and in the open unit disk
\begin{align}
\mathcal{S}_{\tu{S}}  & := \{ \ri x^2 + y^2 < 1 \} \notag \\
& = \{ z \in \compl \colon \smat{-1 & 0\\ 0 & -1} + z \smat{0 & 0 \\ -1 & 0} + \bar z \smat{0 & -1 \\0 & 0} \prec 0\} \label{Schur as LMI region}
\end{align}
or, equivalently, to the existence of a symmetric $P$ satisfying the Lyapunov conditions
\begin{equation*}
\big( P \succ 0, A P + P A^\top \prec 0 \big) \text{ and }
\big( P \succ 0, A P A^\top - P \prec 0 \big);
\end{equation*}
these precise conditions can  be obtained by using $\mathcal{S}_{\tu{H}}$ and $\mathcal{S}_{\tu{S}}$ as LMI regions and applying Fact~\ref{fact:S-stab} to them.
On the other hand, Fact~\ref{fact:S-stab} alone enables considering more general subsets in the complex plane, and Fact~\ref{fact:inters} using their intersections. 

A controller that assigns for a closed-loop system $\dot x/x^+ = A_{\tu{cl}} x$ the eigenvalues of $A_{\tu{cl}}$ in a certain region of the complex plane can effectively enforce meaningful performance specifications since different regions of the complex plane for the eigenvalues of $A_{\tu{cl}}$ correspond to different transient behaviours, as we illustrate in the next example.

\begin{example}
\label{example:wedge}
For suitable parameters $\ell > 0$, $\rho > 0$, $\theta \in (0,\pi/2)$, consider the subset $\mathcal{S}(\ell,\rho,\theta)$ depicted in Fig.~\ref{fig:S(ell,rho,theta)}, left.
$\mathcal{S}(\ell,\rho,\theta) = \{ \cnri x < -\ell \} \cap \{ \cnri x^2 + y^2 < \rho^2 \} \cap \{ \cnri (\cos \theta) |y| < - (\sin \theta) x, x<0 \}$: hence, it guarantees a minimum convergence rate of $\ell$ (halfplane), a maximum natural frequency of $\rho$ (disk) and a minimum damping ratio $\cos \theta$ (cone).
In terms of performance, these correspond to upper bounds on the settling time, the overshoot, the frequency of oscillatory modes and the magnitude of high-frequency poles \cite{chilali1996hInf}\cite[\S 3.3-3.4]{franklin1994feedback}.
By Fact~\ref{fact:inters}, eigenvalues of $A_{\tu{cl}}$ are located in $\mathcal{S}(\ell,\rho,\theta)$ if there exists $P = P^\top \succ 0$ such that
\begin{align*}
& \smat{\ell & 0\\ 0 & -1} \otimes P
+ \Tr\big\{ \smat{1/2 & 0\\ 0 & 0} \otimes (A_{\tu{cl}} P) \big\} \prec 0, \\
& \smat{-\rho & 0\\ 0 & -\rho} \otimes P
+ \Tr\big\{ \smat{0 & 0\\ -1 & 0} \otimes (A_{\tu{cl}} P) \big\} \prec 0, \\
& \Tr\big\{ \smat{\sin \theta & \cos \theta \\ -\cos \theta & \sin \theta} \otimes (A_{\tu{cl}} P) \big\} \prec 0.
\end{align*}
\end{example}
\smallskip

When designing a controller $u = K x$ for $\dot x/x^+ = A x + B u$, one considers a closed-loop matrix $A + B K$ in~\eqref{S-stab cond} and looks for $P = P^\top \succ 0$ and $K$ such that
\begin{equation*}
\alpha \otimes P + \beta \otimes ((A+BK)P) + \beta^\top \otimes (P(A+BK)^\top) \prec 0.
\end{equation*}
This inequality is not linear in $P$ and $K$, so one uses instead%
\begin{subequations}%
\label{sol:model based}%
\begin{align}
& \text{find} & & \hspace*{-4pt} P = P^\top \succ 0, Y \\
& \text{s.~t.} & &  \hspace*{-4pt} \alpha \otimes P + \Tr\big\{ \beta \otimes (AP + BY) \big\} \prec 0 \label{sol:model based:lmi}
\end{align}
\end{subequations}
with \eqref{sol:model based:lmi} now linear in $P$ and $Y$. From the underlying change of variables, $K$ is $Y P^{-1}$.
In the presence of $r$ LMI regions $\mathcal{S}_i$ with data $(\alpha_i, \beta_i)$, $i = 1, \dots , r$, \eqref{sol:model based} extends by Fact~\ref{fact:inters} to
\begin{subequations}%
\label{sol:model based multiple}%
\begin{align}
& \text{find} & & \hspace*{-4pt} P = P^\top \succ 0, Y \\
& \text{s.~t.} & &  \hspace*{-4pt} \alpha_i \otimes P + \Tr\big\{ \beta_i \otimes (AP + BY) \big\} \prec 0, i=1,\dots, r. \label{sol:model based multiple:lmi}
\end{align}
\end{subequations}

\begin{figure}
\centerline{\includegraphics[height=3.5cm]{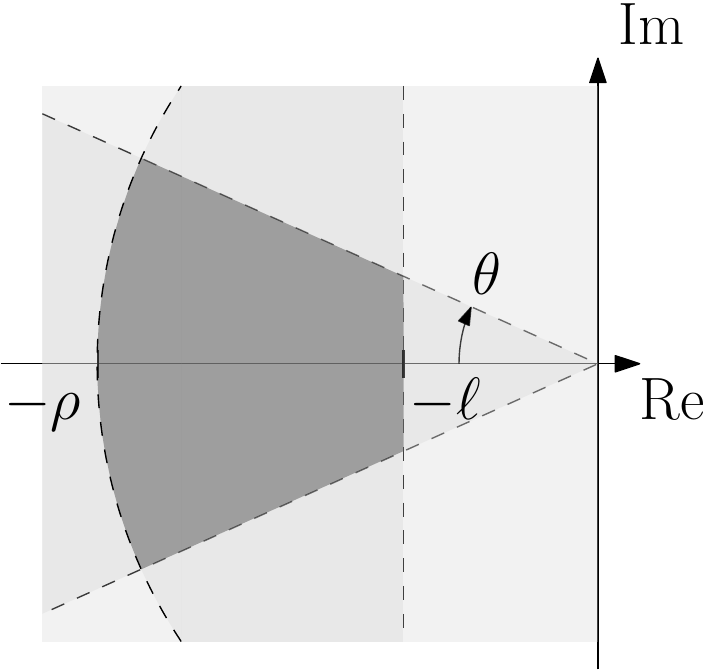}~~\includegraphics[height=3.5cm]{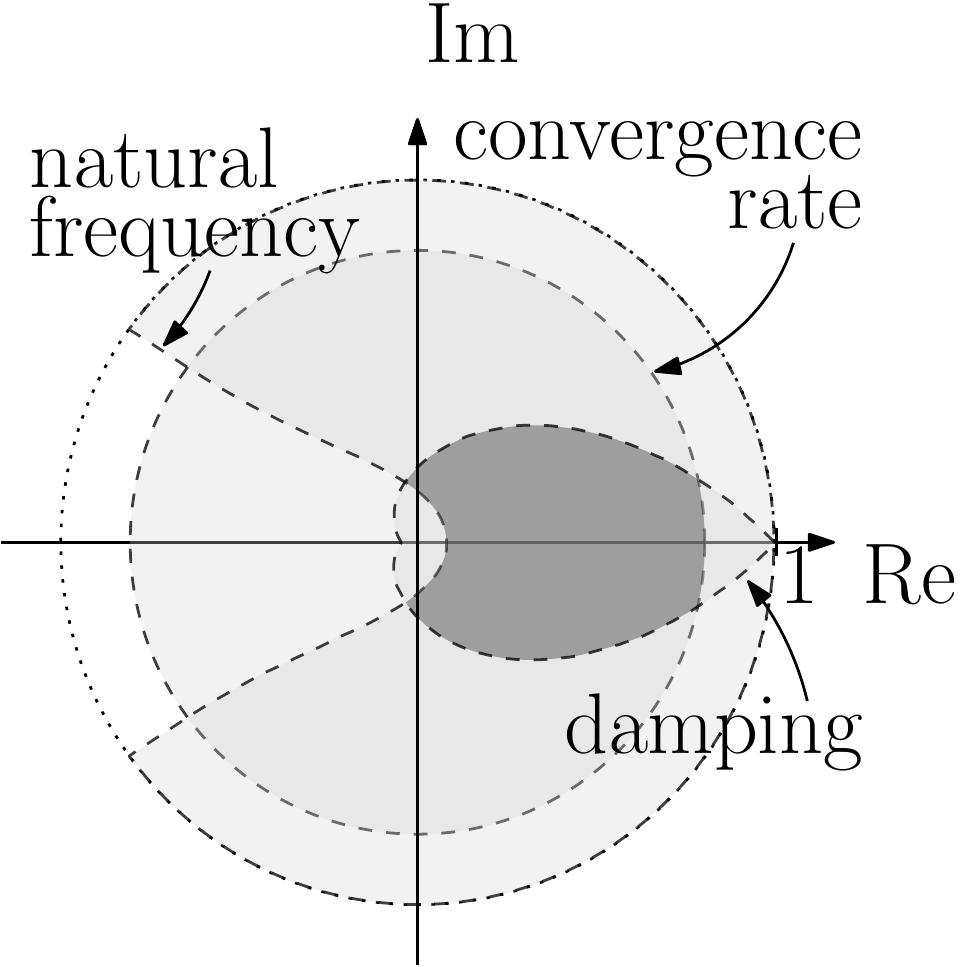}}
\caption{(Left) Region $\mathcal{S}(\ell,\rho,\theta)$ of Example~\ref{example:wedge} enforcing a desirable transient behaviour in continuous time.
(Right) The analogous region in discrete time where loci of constant natural frequency, convergence rate and damping are indicated, see Example~\ref{example:perf dt}.
}
\label{fig:S(ell,rho,theta)}
\end{figure}

\section{Data-driven control by LMI regions}
\label{sec:setting}

\subsection{Problem formulation}

Consider a linear time-invariant system
\begin{equation}
\label{sys}
x^\circ = A_\star x + B_\star u + d
\end{equation}
where $x \in \real^n$ is the state, $u \in \real^m$ is the input, $d \in \real^n$ is the disturbance and $x^\circ$ represents the time derivative $\dot x$ of state $x$ in continuous time or the update $x^+$ of the state $x$ in discrete time. 
For convenience, we call $x^\circ$ state preview. 
The matrices $A_\star$ and $B_\star$ are \emph{unknown} to us, and we rely instead on data collected through an experiment on the system.
Specifically, we apply an input sequence $u(t_0)$, $u(t_1)$, \dots $u(t_{T-1})$ and measure the corresponding state and state preview sequences $x(t_0)$, $x(t_1)$, \dots, $x(t_{T-1})$ and $x^\circ(t_0)$, $x^\circ(t_1)$, \dots, $x^\circ(t_{T-1})$, where the preview sequence in continuous time is the sequence of state derivatives $\dot x(t_0)$, $\dot x(t_1)$, \dots, $\dot x(t_{T-1})$ and in discrete time is the sequence of states $x(t_1)$, $x(t_2)$, \dots, $x(t_{T})$.
During the experiment, a disturbance sequence $d(t_0)$, $d(t_1)$, \dots, $d(t_{T-1})$ acts on system \eqref{sys} and affects the evolution of the system.
This sequence is also \emph{unknown} to us, and due to its influence on the state evolution, we say that data are noisy.
In summary, the measured sequences are collected in the matrices
\begin{subequations}
\label{data matrices}
\begin{align}
U_0 & := \bmat{u(t_0)  & \dots & u(t_{T-1})}\\
X_0 & := \bmat{x(t_0)  & \dots  & x(t_{T-1})}\\
X_1 & := \bmat{x^\circ(t_0)  & \dots  & x^\circ(t_{T-1})} \\
& \phantom{:}=
\begin{cases}
\bmat{\dot x(t_0)  & \dots  & \dot x(t_{T-1})} & \text{ in cont. time}\\
\bmat{x(t_1)  & \dots  & x(t_{T})} & \text{ in discr. time} \notag
\end{cases}
\end{align}
\end{subequations}
and the unknown disturbance sequence in $D_0 := \bmat{d(t_0)  & \dots & d(t_{T-1})}$.
The times $t_0$, $t_1$, \dots, $t_{T-1}$ are taken as the $0, 1, \dots, T-1$ multiples of a certain period $T_{\tu{s}}$; this is a natural choice in discrete time since these times correspond to periodic sampling times, and we adopt the same choice in continuous time as well (although this is not necessary).
Since the data generation mechanism is \eqref{sys}, the data points in the experiment satisfy
\begin{equation*}
X_1 = A_\star X_0 + B_\star U_0 + D_0.
\end{equation*}
As for the disturbance sequence, we know only that, for some $p \in \nat_{\ge 1}$ and some  matrix $\Delta \in \real^{n \times p}$, it belongs to the set
\begin{equation}
\label{set D}
\mathcal{D} := \{ D \in \real^{n \times T} \colon D D^\top \preceq \Delta \Delta^\top \}
\end{equation}
and this corresponds to knowing a bound on the energy of any disturbance sequence, and in particular of $D_0$, which satisfies $D_0 \in \mathcal{D}$.
The pairs $(A,B)$ that could have generated the data points $U_0$, $X_0$, $X_1$ for a disturbance sequence $D \in \mathcal{D}$ correspond to the set
\begin{equation}
\label{set C}
\mathcal{C} := \{ (A,B) \colon X_1 = A X_0 + B U_0 + D, D \in \mathcal{D} \},
\end{equation}
and $\mathcal{C}$ is called the set of \emph{matrices consistent with data}.
Since $D_0 \in \mathcal{D}$, we have $(A_\star, B_\star) \in \mathcal{C}$.
The objective of this work is to design a linear feedback controller $u = K x$ that assigns the eigenvalues of the closed-loop matrix $A_\star + B_\star K$ within a certain subset $\mathcal{R}$ of the complex plane.
Since $A_\star$ and $B_\star$ are unknown, this is achieved by imposing that $K$ assigns the eigenvalues of $A + B K$ within $\mathcal{R}$ for all $(A,B) \in \mathcal{C}$.
As a first step, we address the case of $\mathcal{R}$ given by a single LMI region $\mathcal{S}$ with data $(\alpha, \beta)$ where the objective becomes
\begin{subequations}
\label{probl single LMI region}
\begin{align}
& \text{find} & & \hspace*{-4pt} P = P^\top \succ 0, K \\
& \text{s.~t.} & &  \hspace*{-4pt} \alpha \otimes P + \beta \otimes \big( (A + B K ) P \big) \notag \\
& & & \hspace*{-4pt} \hspace*{2pt} + \beta^\top \otimes \big( P (A + B K)^\top \big) \prec 0 \quad \forall (A,B) \in \mathcal{C}.
\end{align}
\end{subequations}
In words, solving this problem yields a certificate and a controller for $\mathcal{S}$-stability, i.e., the Lyapunov-like matrix $P$ and the gain $K$.
For $\alpha = 0$ and $\beta = 1$, \eqref{probl single LMI region} becomes a classical continuous-time stabilization problem and the corresponding $V(x) = x^\top P x$ is a \emph{common} Lyapunov function \cite[\S II.B]{vanwaarde2020noisy} since it accommodates all matrices $(A,B)$ consistent with data.
Analogously, the matrix $P$ obtained as a solution to a generic \eqref{probl single LMI region} qualifies as a \emph{common} Lyapunov-like matrix.
As a second step, we address the case of $\mathcal{R}$ given by the intersection of $r$ LMI regions $\mathcal{S}_i$, $i=1,\dots, r$, each with data $(\alpha_i, \beta_i)$, that is, $\mathcal{R} := \bigcap_{i=1}^{r} \mathcal{S}_i$.
The objective becomes
\begin{subequations}
\label{probl multiple LMI regions}
\begin{align}
& \text{find}  & & \hspace*{-4pt} P = P^\top \succ 0, K \\
& \text{s.~t.} & &  \hspace*{-4pt} M_{\mathcal{S}_1} \big( (A + B K), P \big) \prec 0, \notag \\
& & & \hspace*{-4pt} \hspace*{2pt} \dots, M_{\mathcal{S}_r} \big( (A + B K), P \big) \prec 0 \quad \forall (A,B) \in \mathcal{C}
\end{align}
\end{subequations}
where the definition of characteristic matrices $M_{\mathcal{S}_1}$, \dots, $M_{\mathcal{S}_r}$ is in \eqref{charact matrix}.
\eqref{probl multiple LMI regions} is the natural extension of~\eqref{probl single LMI region} by taking into account Fact~\ref{fact:inters} for an intersection of LMI regions.

\subsection{Equivalent forms of set $\mathcal{C}$}

The set $\mathcal{C}$ introduced in~\eqref{set C} plays a key role in the developments and in this section we present for it three different forms equivalent to each other.
The first form is%
\begin{subequations}
\label{set C form 1}
\begin{align}
& \mathcal{C} = \{ \bmat{A & B} = Z^\top \colon \bmat{I & Z^\top} \bmat{\mb{C} & \mb{B}^\top\\ \mb{B} & \mb{A} } \bmat{I\\ Z} \preceq 0 \} \label{set C form 1:set only}\\
& \mb{C} : =  - \Delta \Delta^\top + X_1 X_1^\top \label{set C form 1:C}\\
& \mb{B} : =  - \bmat{X_0\\ U_0} X_1^\top , \mb{A} : = \bmat{X_0\\ U_0}\bmat{X_0\\ U_0}^\top, \label{set C form 1:B A}
\end{align}
\end{subequations}
and this form can be obtained with algebraic computations from the definition of $\mathcal{C}$ in~\eqref{set C} by expressing $D$ in \eqref{set C} as $D= X_1 - A X_0 - B U_0$, substituting this $D$ in the condition defining $\mathcal{D}$ in~\eqref{set D}, and collecting $\bmat{I & A & B} = \bmat{I & Z^\top }$ to the left and its transpose to the right.
We make the next assumption on matrix $\smat{X_0\\ U_0}$.
\begin{assumption}
\label{ass:full row rank}
Matrix $\smat{X_0\\ U_0}$ has full row rank.
\end{assumption}
Full row rank of $\smat{X_0\\ U_0}$ is intimately related to persistence of excitation of the input and disturbance sequences, see a detailed discussion in \cite[\S 4.1]{bisoffiArXivPetersen}.
The rank condition can be verified directly from data and when it does not hold, one can typically enforce it simply by collecting more data points, thereby adding columns to $\smat{X_0\\ U_0}$.
By Assumption~\ref{ass:full row rank}, $\mb{A} = \smat{X_0\\ U_0} \smat{X_0\\ U_0}^\top \succ 0$. 
Thanks to $\mb{A} \succ 0$, we have the second form of the set $\mathcal{C}$ as
\begin{subequations}
\label{set C form 2}
\begin{align}
& \mathcal{C} = \big\{ \bmat{A & B} = Z^\top \! \colon (Z-Z_{\tu{c}})^\top \mb{A} (Z - Z_{\tu{c}}) \preceq \mb{Q} \big\} \\
& Z_{\tu{c}} := - \mb{A}^{-1} \mb{B}, \mb{Q} := \mb{B}^\top \mb{A}^{-1} \mb{B} - \mb{C}. \label{set C form 2:Zc and Q}
\end{align}
\end{subequations}
We noted before $\mb{A} \succ 0$; the sign definiteness of $\mb{Q}$ is also a structural property as claimed next.
\begin{lemma}
\label{lemma:signs A and Q}
Under Assumption~\ref{ass:full row rank}, $\mb{A} \succ 0$ and $\mb{Q} \succeq 0$.
\end{lemma}
\begin{proof}
The lemma is the same as \cite[Lemma~1]{bisoffiArXivPetersen}, so the proof is omitted.
\end{proof}

With Lemma~\ref{lemma:signs A and Q}, we can give the third form of $\mathcal{C}$ as
\begin{equation}
\label{set C form 3}
\mathcal{C} = \{ Z_{\tu{c}} + \mb{A}^{-1/2} \Upsilon \mb{Q}^{1/2} \colon \Upsilon^\top \Upsilon \preceq I \}.
\end{equation}
The fact that the set $\mathcal{C}$ in \eqref{set C form 2} rewrites equivalently as in \eqref{set C form 3} is straightforward for $\mb{A} \succ 0$ and $\mb{Q} \succ 0$; it is less so for $\mb{A} \succ 0$ and $\mb{Q} \succeq 0$ and the proof for this case is in \cite[Prop.~1]{bisoffiArXivPetersen}.
The third form of $\mathcal{C}$ in~\eqref{set C form 3} is the one we typically need to obtain our main results in the sequel.

\begin{remark}
\label{remark:different dist bound}
Instead of~\eqref{set D}, one can consider for the disturbance sequence the bound given by
\begin{equation*}
\mathcal{D} := \{ D \in \real^{n \times T} \colon \bmat{I & D} \bmat{R & S^\top \\ S & Q} \bmat{I \\ D^\top} \preceq 0 \}
\end{equation*}
with matrices $R$ and $Q$ symmetric and $Q \succ 0$ \cite{berberich2019robust,vanwaarde2020noisy,berberich2020combining}. Because of $Q \succ 0$, $D$ cannot be too ``large''; so, as in~\eqref{set D}, the knowledge of $\mathcal{D}$ corresponds to knowing a bound on the energy of any disturbance sequence.
Moreover, one can obtain a set $\mathcal{C}$ analogous to~\eqref{set C form 1}, namely, $\mathcal{C}$ as in~\eqref{set C form 1:set only} with, instead of \eqref{set C form 1:C}-\eqref{set C form 1:B A},
\begin{align*}
& \mb{C} : =  R + X_1 S + S^\top X_1^\top + X_1 Q X_1^\top\\
& \mb{B} : =  - \bmat{X_0\\ U_0} (S +  Q X_1^\top) , \mb{A} : = \bmat{X_0\\ U_0}  Q \bmat{X_0\\ U_0}^\top,
\end{align*}
and still conclude from Assumption~\ref{ass:full row rank} and $Q\succ 0$ that $\mb{A} = \smat{X_0\\ U_0}  Q^{1/2} Q^{1/2} \smat{X_0\\ U_0}^\top \succ 0$.
Finally, Lemma~\ref{lemma:signs A and Q} remains valid for the different expressions of $\mb{A}$ and, consequently, of $\mb{Q}$ in~\eqref{set C form 2:Zc and Q} and it can be proven with the very same proof strategy of \cite[Lemma~1]{bisoffiArXivPetersen} as long as one substitutes the expression of $\mb{Q}_{\tu{p}}$ there with
\begin{equation*}
\mb{Q}_{\tu{p} } := Q^{1/2}  \smat{X_0\\ U_0}^\top \left( \smat{X_0\\ U_0}  Q \smat{X_0\\ U_0}^\top \right)^{-1} \smat{X_0\\ U_0} Q^{1/2}.
\end{equation*}
The rest of our result hold identically.
\end{remark}

\section{Sufficient condition for generic LMI regions}
\label{sec:suff cond}

In this section we consider generic LMI regions and, for them, look for data-based counterparts of Facts~\ref{fact:S-stab} and \ref{fact:inters}, which will result in the sufficient conditions of Theorem~\ref{thm:suff cond} given next and of Corollary~\ref{cor:suff cond} given later.

\begin{theorem}
\label{thm:suff cond}
Let Assumption~\ref{ass:full row rank} hold and $\mathcal{S}$ be an LMI region with data $(\alpha,\beta)$. \eqref{probl single LMI region} is feasible if the next program is feasible
\begingroup
\setlength{\arraycolsep}{1.5pt}
\begin{subequations}
\label{sol:single LMI region}
\begin{align}
& \hspace{-5pt} \text{find} & & \hspace{-5pt} P=P^\top \succ 0, Y \\
& \hspace{-5pt} \text{s.~t.} & & \hspace{-10pt}
\bmat{
\left\{
\begin{aligned}
& (\beta \beta^\top) \otimes \mb{Q} +  \alpha \otimes P \\
& \hspace*{12pt} + \Tr\big\{\beta \otimes \big( Z_{\tu{c}}^\top \smat{P \\ Y} \big) \big\}
\end{aligned}
\right\}
& \star \\
I_s \otimes  \smat{P \\ Y} &  - I_s \otimes \mb{A} } \prec 0 
\label{sol:single LMI region:lmi}
\end{align}
\end{subequations}%
\endgroup
If \eqref{sol:single LMI region} is feasible, the controller gain in~\eqref{probl single LMI region} is $K = Y P^{-1}$.
\end{theorem}
\begin{proof}
The proof shows how to enforce the condition of Fact~\ref{fact:S-stab} for all matrices $(A,B) \in \mathcal{C}$. Apply Schur complement to~\eqref{sol:single LMI region:lmi} and obtain equivalently, since $\mb{A} \succ 0$,
\begin{align*}
& (\beta \beta^\top) \otimes \mb{Q} +  \alpha \otimes P + \Tr\big\{\beta \otimes \big( Z_{\tu{c}}^\top \smat{P \\ Y} \big) \big\} \\
& \hspace*{5pt}  + 
\big( I_s \otimes  \smat{P \\ Y}^\top \big) 
\big( I_s \otimes \mb{A}^{-1} \big)
\big( I_s \otimes  \smat{P \\ Y} \big)
\prec 0.
\end{align*}
Equivalently, multiply both sides by any $\lambda >0$ and then replace $\lambda P$, $\lambda Y$ with $P$, $Y$ to obtain
\begin{align}
& \alpha \otimes P + \Tr\big\{ \beta \otimes \big( Z_{\tu{c}}^\top \smat{P \\ Y} \big) \big\} \notag \\
& + \frac{1}{\lambda} \Big( I_s \otimes  \big( \smat{P \\ Y}^\top \mb{A}^{-1/2}  \big) \Big) \Big( I_s \otimes  \big( \mb{A}^{-1/2} \smat{P \\ Y} \big) \Big) \notag \\
& + \lambda \Big( \beta \otimes \mb{Q}^{1/2} \Big) \Big(\beta^\top \otimes \mb{Q}^{1/2}\Big)
\prec 0. \label{sol:sLMIr:lmi after schur}
\end{align}
We use the so-called Petersen's lemma \cite{petersen1987stabilization} in the version reported in \cite[Fact~1]{bisoffiArXivPetersen}.
The existence of $\lambda > 0$ such that \eqref{sol:sLMIr:lmi after schur} holds is equivalent by \cite[Fact~1]{bisoffiArXivPetersen} to
\begin{align}
& \alpha \otimes P +  \Tr \big\{ \beta \otimes \big( Z_{\tu{c}}^\top \smat{P \\ Y} \big) \big\}  \notag \\
& + \Tr \Big\{ \big( \beta \otimes \mb{Q}^{1/2} \big) \mb{Y}^\top \Big( I_s \otimes  \big( \mb{A}^{-1/2} \smat{P \\ Y} \big) \Big) \Big\} \notag\\
& \prec 0 \quad \forall \mb{Y} \colon \mb{Y}^\top \mb{Y} \preceq I_{sn} \label{sol:sLMIr:lmi full Y}
\end{align}
With $\Upsilon$ as in~\eqref{set C form 3}, i.e., $\Upsilon^\top \Upsilon \preceq I_n$, consider the block-diagonal matrix $\mb{Y} = I_s \otimes \Upsilon$ with $s$ blocks. 
$\mb{Y}^\top \mb{Y} = \smat{\Upsilon^\top \Upsilon &  & 0\\  & \rotatebox{-45}{\resizebox{8pt}{!}{\ldots}} & \\ 0 &  & \Upsilon^\top \Upsilon } \preceq \smat{I_n & & 0\\  & \rotatebox{-45}{\resizebox{8pt}{!}{\ldots}} & \\ 0 & & I_n} = I_{sn}$ so \eqref{sol:sLMIr:lmi full Y} implies for the selected block-diagonal $\mb{Y}$ that
\begin{align}
& 0 \succ \alpha \otimes P + \Tr \big\{ \beta \otimes \big( Z_{\tu{c}}^\top \smat{P \\ Y} \big) \big\}  \notag \\
& + \Tr \Big\{ \big( \beta \otimes \mb{Q}^{1/2} \big) \big( I_s \otimes \Upsilon^\top \big) \Big( I_s \otimes  \big( \mb{A}^{-1/2} \smat{P \\ Y} \big) \Big) \Big\} \notag\\
& = \alpha \otimes P + \Tr \big\{ \beta \otimes \big( Z_{\tu{c}}^\top \smat{P \\ Y} \big) \big\} \notag \\
& \quad + \Tr \Big\{  \beta \otimes \Big( \mb{Q}^{1/2}  \Upsilon^\top \mb{A}^{-1/2} \smat{P \\ Y} \Big) \Big\} \notag\\
& = \alpha \otimes P + \Tr \Big\{  \beta \otimes \Big( \big( Z_{\tu{c}} + \mb{A}^{-1/2} \Upsilon \mb{Q}^{1/2} \big)^\top \smat{P \\ Y}  \Big) \Big\} \notag \\
& \hspace*{20pt} \forall \Upsilon \colon \Upsilon^\top \Upsilon \preceq I_n.
\label{sol:sLMIr:lmi implies}
\end{align}
Equivalently, we have by~\eqref{set C form 3} that
\begin{align}
& 0 \succ \alpha \otimes P + \Tr \big\{  \beta \otimes \big( \smat{A & B} \smat{P \\ Y}  \big) \big\} \quad \forall (A,B) \in \mathcal{C}. \label{sol:sLMIr:lmi (A,B)}
\end{align}
In summary, feasibility of \eqref{sol:single LMI region} implies feasibility of
\begin{equation*}
\text{find } P = P^\top \succ 0, Y \text{ subject to } \eqref{sol:sLMIr:lmi (A,B)}.
\end{equation*}
Feasibility of this problem is equivalent to feasibility of \eqref{probl single LMI region} by the standard change of variables given by $Y=KP$.
\end{proof}
\smallskip

The feasibility program in~\eqref{sol:single LMI region} is convenient since the constraint \eqref{sol:single LMI region:lmi} is a linear matrix inequality in the decision variables $P$, $Y$.
With respect to the model-based condition in Fact~\ref{fact:S-stab}, Theorem~\ref{thm:suff cond} no longer gives a necessary and sufficient condition because, from~\eqref{sol:sLMIr:lmi full Y} to \eqref{sol:sLMIr:lmi implies} in the proof, we used that for matrices $\mb{D} = \mb{D}^\top$, $\mb{E}$, $\mb{G}$
\begin{subequations}
\begin{equation}
\label{for all full}
0 \succ \mb{D} + \mb{E} \mb{F} \mb{G} + \mb{G}^\top \mb{F}^\top \mb{E}^\top \quad \forall \mb{F} \colon \mb{F}^\top \mb{F} \preceq I
\end{equation}
implies
\begin{equation}
\label{for all block}
0 \succ \mb{D} + \mb{E} \big( I \otimes \mb{f} \big) \mb{G} + \mb{G}^\top \big( I \otimes \mb{f}^\top \big) \mb{E}^\top \,\, \forall \mb{f} \colon \mb{f}^\top \mb{f} \preceq I,
\end{equation}
\end{subequations}
but is not implied by it in general\footnote{
Take $\mb{D} = - I$, $\mb{E} = \smat{1 & 0\\ 1 & 0}$, $\mb{G} = \smat{0 & 0\\ 1 & 1}$.
$\mb{D} + \mb{E} \smat{\mb{f} & 0\\ 0 & \mb{f} } \mb{G} + \mb{G}^\top \smat{\mb{f} & 0\\ 0 & \mb{f}}  \mb{E}^\top = \mb{D} \prec 0 $ for all $\mb{f} \in \real$.
On the other hand, take $\mb{F} = \smat{ 0 & 1\\ 0 & 0}$, which satisfies $\mb{F}^\top \mb{F} \preceq I$; for such $\mb{F}$, $\mb{D} + \mb{E} \mb{F} \mb{G} + \mb{G}^\top \mb{F}^\top \mb{E}^\top = \smat{1 & 2 \\ 2 & 1}$, which is not negative definite. 
This shows that \eqref{for all block} does not imply \eqref{for all full}.}.
The larger the number $s$ of blocks $\mb{f}$ on the diagonal is, the sparser the matrix $I \otimes \mb{f}$ (of unit norm) is, the more conservative it is to replace that matrix with the full matrix $\mb{F}$ (of unit norm).
Hence, the chances of feasibility decrease with the dimension $s$ of the LMI region considered in Theorem~\ref{thm:suff cond}.
From Theorem~\ref{thm:suff cond}, which is the data-based counterpart of Fact~\ref{fact:S-stab}, we obtain the next data-based counterpart of Fact~\ref{fact:inters}.
\begin{corollary}
\label{cor:suff cond}
Let Assumption~\ref{ass:full row rank} hold and $\mathcal{S}_i$, for $i = 1, \dots , r$, be an LMI region with data $(\alpha_i, \beta_i)$. 
\eqref{probl multiple LMI regions} is feasible if the next program is feasible
\begin{subequations}
\label{sol:multiple LMI regions}
\begin{align}
& \hspace{-5pt} \text{find} & & \hspace{-5pt} P=P^\top \succ 0, Y\\
& \hspace{-5pt} \text{s.~t.} & & \hspace{-10pt} 
\bmat{
\left\{
\begin{aligned}
& (\beta_i \beta_i^\top) \otimes \mb{Q} +  \alpha_i \otimes P \\
& \hspace*{7pt} + \Tr\big\{\beta_i \otimes \big( Z_{\tu{c}}^\top \smat{P \\ Y} \big) \big\}
\end{aligned}
\right\}
& \star \\
I_s \otimes  \smat{P \\ Y}  &  - I_s \otimes \mb{A} } \prec 0  \notag\\
& & & \text{ for } i = 1,\dots, r.
\label{sol:multiple LMI regions:lmi}
\end{align}
\end{subequations}%
If \eqref{sol:multiple LMI regions} is feasible, the controller gain in~\eqref{probl multiple LMI regions} is $K = Y P^{-1}$.
\end{corollary}
\begin{proof}
\eqref{probl multiple LMI regions} is equivalent, by definition of characteristic matrix in~\eqref{charact matrix}, to
\begingroup
\thinmuskip=0.3mu plus 1mu
\medmuskip=1mu plus 2mu
\thickmuskip=1.5mu plus 3mu
\begin{align*}
& \text{find} & & \hspace*{-7pt} P = P^\top \succ 0, K \\
& & &  \hspace*{-7pt} \alpha_1 \otimes P + \Tr\big\{ \beta_1 \otimes \big( (A + B K ) P \big) \big\} \prec 0 \quad \forall (A,B) \in \mathcal{C} \\
& & & \hspace*{-7pt} \hspace{155pt} \rotatebox{90}{\resizebox{10pt}{!}{\ldots}} \\
& & & \hspace*{-7pt} \alpha_r \otimes P + \Tr\big\{ \beta_r \otimes \big( (A + B K ) P \big) \big\} \prec 0 \quad \forall (A,B) \in \mathcal{C}.
\end{align*}
\endgroup
\eqref{sol:multiple LMI regions} implies feasibility of this program by Theorem~\ref{thm:suff cond}.
\end{proof}

\section{Necessary and sufficient condition for special LMI regions}
\label{sec:nec and suff cond}

We have seen in Section~\ref{sec:suff cond} sufficient conditions for data-driven stabilization within generic LMI regions and their intersections.
In this section we show that necessary and sufficient conditions can be found for special LMI regions and their intersections.
We will show that these are vertical halfplanes, disks centered on the real axis and intersections of such halfplanes and disks, and that these regions can inner-approximate subsets of the complex plane of practical interest.

In Section~\ref{sec:suff cond}, the source of conservatism leading to sufficient conditions was substituting the block diagonal uncertainty $I \otimes \mb{f}$ in~\eqref{for all block} with the full uncertainty $\mb{F}$ in~\eqref{for all full}.
A special case when no conservatism is introduced is when $I \otimes \mb{f} = 1 \otimes \mb{f}$, and this occurs for $\beta$ of rank $1$ as we now show.
By considering at~\eqref{sol:sLMIr:lmi implies} in the proof of Theorem~\ref{thm:suff cond} and setting $\mb{D}:=\alpha \otimes P + \Tr \big\{ \beta \otimes \big( Z_{\tu{c}}^\top \smat{P \\ Y} \big) \big\}$, we have
\begin{align}
& 0 \succ \mb{D} + \Tr \Big\{  \beta \otimes \Big( \mb{Q}^{1/2}  \Upsilon^\top \mb{A}^{-1/2} \smat{P \\ Y} \Big) \Big\} \notag\\
& \hspace*{140pt} \forall \Upsilon \colon \Upsilon^\top \Upsilon \preceq I;
\label{point with conservatism}
\end{align}
hence, if $\beta = \eta \cdot 1 \cdot \gamma^\top$ for some vectors $\eta$ and $\gamma$ in $\real^s$, \eqref{point with conservatism} becomes that for all $\Upsilon$ such that $\Upsilon^\top \Upsilon \preceq I$,
\begin{align*}
& 0 \succ \mb{D} + \Tr \big\{ (\eta \otimes \mb{Q}^{1/2})(1 \otimes \Upsilon^\top)(\gamma^\top \otimes \mb{A}^{-1/2} \smat{P \\ Y}) \big\},
\end{align*}
where $1 \otimes \Upsilon^\top$ appears as we intended to show.
This discussion is summarized in the next assumption.
\begin{assumption}
\label{ass:rank 1}
Let $\mathcal{S}$ be an LMI region with data $(\alpha,\beta)$.
For $\beta \in \real^{s \times s}$, there exist $\eta$ and $\gamma$ in $\real^s$ such that $\beta= \eta \gamma^\top$, i.e., 
$\smat{\beta_{11} & \dots & \beta_{1s} \\ \rotatebox[origin=c]{90}{{\resizebox{8pt}{!}{\ldots}}}  &  & \rotatebox[origin=c]{90}{{\resizebox{8pt}{!}{\ldots}}} \\ \beta_{s1} & \dots & \beta_{ss} } 
= 
\smat{\eta_1 \gamma_1 & \dots & \eta_1 \gamma_s \\ \rotatebox[origin=c]{90}{{\resizebox{8pt}{!}{\ldots}}}  &  & \rotatebox[origin=c]{90}{{\resizebox{8pt}{!}{\ldots}}} \\ 
\eta_s \gamma_1 & \dots & \eta_s \gamma_s}$.
\end{assumption}

A simple exemplification of this assumption follows.

\begin{example}
\label{example:assumption rank 1}
An open disk with center $(x_{\tu{d}},0)$ and radius $r_{\tu{d}}>0$ has data $(\alpha_{\tu{d}},\beta_{\tu{d}})$ with $\alpha_{\tu{d}} :=\smat{-r_{\tu{d}} & x_{\tu{d}} \\  x_{\tu{d}}  & -r_{\tu{d}}}$ and $\beta_{\tu{d}}:=\smat{0 & 0\\ -1 & 0}$, and satisfies Assumption~\ref{ass:rank 1} with $\eta_{\tu{d}}:=\smat{0\\ -1}$ and $\gamma_{\tu{d}}:=\smat{1\\ 0}$.
\end{example}
 
For an LMI region $\mathcal{S}$ satisfying Assumption~\ref{ass:rank 1}, the program \eqref{probl single LMI region} is reformulated \emph{equivalently} as in the next result.
\begin{theorem}
\label{thm:nec and suff cond}
Let Assumption~\ref{ass:full row rank} hold and $\mathcal{S}$ be an LMI region with data $(\alpha, \beta)$ satisfying Assumption~\ref{ass:rank 1}, i.e., $\beta = \eta \gamma^\top$ for some $\eta$ and $\gamma$ in $\real^s$. 
Then, \eqref{probl single LMI region} is feasible if and only if the next program is feasible 
\begingroup
\setlength{\arraycolsep}{1.5pt}
\thinmuskip=0.3mu plus 1mu
\medmuskip=1mu plus 2mu
\thickmuskip=1.5mu plus 3mu
\begin{subequations}
\label{sol rank 1:single LMI region}
\begin{align}
& \hspace{-5pt} \text{find} & & \hspace{-5pt} P=P^\top \succ 0, Y\\
& \hspace{-5pt} \text{s.~t.} & & \hspace{-10pt}
\bmat{
\left\{
\begin{aligned}
& (\eta \otimes I_n) \mb{Q} (\eta^\top \otimes I_n) + \alpha \otimes P    \\
& \hspace*{3pt} + \Tr\big\{  (\eta \otimes I_n) Z_{\tu{c}}^\top \big(\gamma^\top \otimes \smat{P \\ Y} \big) \big\}
\end{aligned}
\right\}
& \star \\
\gamma^\top \otimes  \smat{P \\ Y}  &  -\mb{A}} \prec 0 \label{sol rank 1:single LMI region:lmi}
\end{align}
\end{subequations}%
\endgroup
If \eqref{sol rank 1:single LMI region} is feasible, the controller gain in~\eqref{probl single LMI region} is $K = Y P^{-1}$.
\end{theorem}
\begin{proof}
By Schur complement and $\mb{A} \succ 0$ by Assumption~\ref{ass:full row rank}, \eqref{sol rank 1:single LMI region:lmi} is equivalent to
\begin{align*}
& 0 \succ \alpha \otimes P + \Tr\big\{  (\eta \otimes I_n) Z_{\tu{c}}^\top \big(\gamma^\top \otimes \smat{P \\ Y} \big) \big\} \\
& \hspace*{-2pt} + (\eta \otimes I_n) \mb{Q} (\eta^\top \otimes I_n) + \big( \gamma \otimes  \smat{P \\ Y}^\top  \big) \mb{A}^{-1} \big( \gamma^\top \otimes \smat{P \\ Y} \big)
\end{align*}
Equivalently, multiply by any $\lambda > 0$ both sides and then replace $\lambda P$, $\lambda Y$ with $P$, $Y$ to obtain
\begin{align}
& 0 \succ \alpha \otimes P + \Tr\big\{ (\eta \otimes I) Z_{\tu{c}}^\top \big(\gamma^\top \otimes \smat{P \\ Y} \big) \big\} \notag \\
& + \lambda ( \eta \otimes I_n ) \mb{Q}^{1/2} \mb{Q}^{1/2} ( \eta^\top \otimes I_n ) \notag \\
& + \frac{1}{\lambda} \big( \gamma \otimes  \smat{P \\ Y}^\top \big) \mb{A}^{-1/2} \mb{A}^{-1/2}   \big( \gamma^\top \otimes  \smat{P \\ Y} \big) . \label{sol rank 1:sLMIr:after Schur}
\end{align}
By Petersen's lemma \cite{petersen1987stabilization} in the version \cite[Fact~1]{bisoffiArXivPetersen}, the existence of $\lambda > 0$ such that \eqref{sol rank 1:sLMIr:after Schur} holds is equivalent to
\begin{align*}
& 0 \succ \alpha \otimes P + \Tr\big\{ (\eta \otimes I_n) Z_{\tu{c}}^\top \big(\gamma^\top \otimes \smat{P \\ Y} \big) \big\} \notag \\
& + \Tr \Big\{ ( \eta \otimes I_n ) \mb{Q}^{1/2} \Upsilon^\top  \mb{A}^{-1/2}   \big( \gamma^\top \otimes \smat{P \\ Y} \big) \Big\} \\
& = \alpha \otimes P + \Tr\Big\{ (\eta \otimes I_n) \Big( Z_{\tu{c}}^\top \\
& \hspace*{20pt} + \mb{Q}^{1/2} \Upsilon^\top  \mb{A}^{-1/2} \Big) \big( \gamma^\top \otimes \smat{P \\ Y} \big) \Big\} \,\, \forall \Upsilon \colon \Upsilon^\top \Upsilon \preceq I_n.
\end{align*}
This condition is equivalent, by~\eqref{set C form 3}, to
\begin{align*}
& 0 \succ  \alpha \otimes P + \Tr\Big\{ (\eta \otimes I_n) \smat{A & B} \big(\gamma^\top \otimes \smat{P \\ Y} \big) \Big\} \\
& = \alpha \otimes P + \Tr\Big\{ (\eta \otimes I_n) (1 \otimes \smat{A & B} ) \big(\gamma^\top \otimes \smat{P \\ Y} \big) \Big\} \\
& = \alpha \otimes P + \Tr\Big\{ (\eta \gamma^\top) \otimes \big(\smat{A & B} \smat{P \\ Y} \big) \Big\} \\
& = \alpha \otimes P + \Tr\Big\{ \beta \otimes \big(\smat{A & B} \smat{P \\ Y} \big) \Big\}   \,\, \forall (A,B) \in \mathcal{C}.
\end{align*}
To find $P = P^\top \succ 0$ and $Y$ subject to this condition is equivalent to \eqref{probl single LMI region} by the standard change of variables $Y = K P$.
\end{proof}

Parallel to Corollary~\ref{cor:suff cond}, we have the next result for the intersection of LMI regions satisfying Assumption~\ref{ass:rank 1}. 
\begin{corollary}
\label{cor:nec and suff cond}
Let Assumption~\ref{ass:full row rank} hold and $\mathcal{S}_i$, for $i = 1, \dots , r$, be an LMI region with data $(\alpha_i, \beta_i)$ satisfying Assumption~\ref{ass:rank 1}, i.e., $\beta_i = \eta_i \gamma_i^\top$ for some $\eta_i$ and $\gamma_i$ in $\real^s$. \eqref{probl multiple LMI regions} is feasible if and only if the next program is feasible
\begin{subequations}
\label{sol nec and suff:multiple LMI regions}
\begin{align}
& \hspace{-5pt} \text{find} & & \hspace{-5pt} P=P^\top \succ 0, Y \\
& \hspace{-5pt} \text{s.~t.} & & \hspace{-10pt} 
\bmat{
\left\{
\begin{aligned}
& (\eta_i \otimes I_n) \mb{Q} (\eta_i^\top \otimes I_n) + \alpha_i \otimes P    \\
& \hspace*{3pt} + \Tr\big\{  (\eta_i \otimes I_n) Z_{\tu{c}}^\top \big(\gamma_i^\top \otimes \smat{P \\ Y} \big) \big\}
\end{aligned}
\right\}
& \star \\
\gamma_i^\top \otimes  \smat{P \\ Y}  &  -\mb{A}} \prec 0  \notag\\
& & & \text{ for } i = 1,\dots, r.
\end{align}
\end{subequations}%
If \eqref{sol nec and suff:multiple LMI regions} is feasible, the controller gain in~\eqref{probl multiple LMI regions} is $K = Y P^{-1}$.
\end{corollary}

Motivated by Theorem~\ref{thm:nec and suff cond} and Corollary~\ref{cor:nec and suff cond}, we examine the subsets of the complex plane to which LMI regions satisfying Assumption~\ref{ass:rank 1} give rise.
A generic LMI region with $s=1$ and data $(\alpha,\beta)$ has trivially $\beta$ of rank $1$, and is $\mathcal{S} = \{ \cnri \alpha_{11} + \beta_{11} (z + \bar z) = \alpha_{11} + x 2 \beta_{11} < 0 \}$, which can express vertical halfplanes, besides being possibly $\emptyset$ or $\compl$.
For $s=2$, the subsets of $\compl$ that can be expressed with $\beta$ of rank $1$ are determined in the next lemma.
\begin{lemma}
\label{lemma:regions s=2 beta rank 1}
Let $\mathcal{S}$ be an LMI region with $s=2$ and data $(\alpha,\beta)$ satisfying Assumption~\ref{ass:rank 1}, i.e., $\beta = \eta \gamma^\top$ for some $\eta$ and $\beta$ in $\real^2$.
Then, a nontrivial $\mathcal{S}$ (i.e., different from $\emptyset$ and $\compl$) can only be: a vertical strip, a vertical halfplane, a disk centered on the real axis or an intersection of the last two.
\end{lemma}
\begin{proof}
Specialize the expression in~\eqref{generic LMI region s=2} of a generic LMI region with $s=2$ for the $\mathcal{S}$ considered in the statement with $\eta:=\smat{\eta_1\\ \eta_2}$, $\gamma:=\smat{\gamma_1 \\ \gamma_2}$ and $\beta:=\smat{\eta_1\gamma_1 &  \eta_1\gamma_2\\ \eta_2\gamma_1 &  \eta_2\gamma_2}$.
Then, $\mathcal{S}$ rewrites after some algebraic computations as the set of points $x + j y \in \compl$ such that
\begin{subequations}
\label{LMI region s=2 rank 1 after subst}
\begin{align}
& 0 > \alpha_{11} + 2 x (\eta_1 \gamma_1 ) \label{LMI region s=2 rank 1 after subst:affine}\\
& 0 < \alpha_{11} \alpha_{22} - \alpha_{12}^2 + 2x \big(\alpha_{11} \eta_2 \gamma_2 + \alpha_{22} \eta_1 \gamma_1 \notag \\
& - \alpha_{12}(\eta_1 \gamma_2 + \eta_2 \gamma_1) \big) - ( x^2 + y^2) (\eta_1 \gamma_2 -\eta_2 \gamma_1)^2 . \label{LMI region s=2 rank 1 after subst:quadratic}
\end{align}
\end{subequations}
If $\eta_1 \gamma_2 -\eta_2 \gamma_1 = 0$, $\mathcal{S}$ is constituted from~\eqref{LMI region s=2 rank 1 after subst} by two affine inequalities in the real part $x$, hence a nontrivial $\mathcal{S}$ can be a vertical halfplane or a vertical strip depending on the values of $\alpha$, $\eta$, $\gamma$.
Otherwise, divide by $(\eta_1 \gamma_2 -\eta_2 \gamma_1)^2$ in~\eqref{LMI region s=2 rank 1 after subst:quadratic} and complete the square, which can be done for each $\alpha$, $\eta$, $\gamma$ with $\eta_1 \gamma_2 -\eta_2 \gamma_1 \neq 0$; after some computations, \eqref{LMI region s=2 rank 1 after subst} rewrites equivalently as
\begin{align*}
& 0 > \alpha_{11} + 2 x (\eta_1 \gamma_1 ), \sigma> (x_0 - x)^2 + y^2  \\
& x_0:= \big( \alpha_{11} \eta_2 \gamma_2 + \alpha_{22} \eta_1 \gamma_1 - \alpha_{12}(\eta_1 \gamma_2 + \eta_2 \gamma_1) \big) \notag  \\
&  \hspace*{35pt} \cdot \big( \eta_1 \gamma_2 -\eta_2 \gamma_1 \big)^{-2}\\
&  \sigma := \big( \alpha_{11}^2 \eta_2^2 \gamma_2^2 + \alpha_{22}^2 \eta_1^2 \gamma_1^2 + \alpha_{12}^2 4 \eta_1 \eta_2 \gamma_1 \gamma_2  \notag \\
& + \alpha_{11} \alpha_{22} (\eta_1^2 \gamma_2^2 + \eta_2^2 \gamma_1^2) -2 \alpha_{11} \alpha_{12} \eta_2 \gamma_2 (\eta_1 \gamma_2 + \eta_2 \gamma_1) \notag \\
& -2 \alpha_{12} \alpha_{22} \eta_1 \gamma_1 (\eta_1 \gamma_2 + \eta_2 \gamma_1)   \big) \cdot \big( \eta_1 \gamma_2 -\eta_2 \gamma_1 \big)^{-4}  
\end{align*}
If $\sigma$ is nonpositive, $\mathcal{S}$ is an empty set. Otherwise, a nontrivial $\mathcal{S}$ can be a disk centered on the real axis or an intersection of it with a halfplane. 
\end{proof}

Let us exemplify Lemma~\ref{lemma:regions s=2 beta rank 1} and Theorem~\ref{thm:nec and suff cond} on the important special cases of Hurwitz and Schur stability.

\begin{example}
We have shown in~\eqref{Hurwitz as LMI region} how the condition for Hurwitz stability can be expressed through an LMI region $\mathcal{S}_{\tu{H}}$ with $\alpha = 0$ and $\beta = 1$.
Such LMI region satisfies Assumption~\ref{ass:rank 1} trivially with, e.g., $\eta = \gamma =1$.
Theorem~\ref{thm:nec and suff cond} claims then that, under Assumption~\ref{ass:full row rank}, \eqref{probl single LMI region} is feasible if and only if there exist $P = P^\top \succ 0$ and $Y$ such that
\begin{equation*}
\bmat{
\Tr\big\{  Z_{\tu{c}}^\top \big(\smat{P \\ Y} \big) \big\} + \mb{Q} & \star \\
\smat{P \\ Y}  &  -\mb{A}} \prec 0.
\end{equation*}
This claim is precisely the same as \cite[Thm.~2]{bisoffiArXivPetersen} by taking into account \cite[Remark~3 and Eq.~(38)]{bisoffiArXivPetersen}.
\end{example}

\begin{example}
We have shown in~\eqref{Schur as LMI region} how the condition for Schur stability can be expressed through an LMI region $\mathcal{S}_{\tu{S}}$ with $\alpha = \smat{-1 & 0\\ 0 & -1}$ and $\beta = \smat{0 & 0\\ -1 & 0}$.
As shown in Example~\ref{example:assumption rank 1}, such LMI region satisfies Assumption~\ref{ass:rank 1} with, e.g., $\eta =\smat{0\\ -1}$ and $\gamma =\smat{1\\ 0}$.
Theorem~\ref{thm:nec and suff cond} claims then that, under Assumption~\ref{ass:full row rank}, \eqref{probl single LMI region} is feasible if and only if there exist $P = P^\top \succ 0$ and $Y$ such that, after some computations,
\begingroup
\setlength{\arraycolsep}{1.2pt}
\thinmuskip=.1mu plus 1mu
\medmuskip=.2mu plus 2mu
\thickmuskip=.3mu plus 3mu
\begin{align*}
\smat{
-P & -\smat{P\\Y}^\top Z_{\tu{c}} & \smat{P\\Y}^\top \\
- Z_{\tu{c}}^\top\smat{P\\ Y} & -P + \mb{Q} & 0\\
\smat{P\\ Y} & 0 & -\mb{A}} \prec 0 \iff 
\smat{
-P + \mb{Q} &  - Z_{\tu{c}}^\top\smat{P\\ Y} & 0\\ 
-\smat{P\\Y}^\top Z_{\tu{c}} &  -P & \smat{P\\Y}^\top \\
0 & \smat{P\\ Y} & -\mb{A}} \prec 0. 
\end{align*}
\endgroup
This claim is precisely the same as \cite[Thm.~1]{bisoffiArXivPetersen} by taking into account \cite[Remark~3 and Eq.~(37)]{bisoffiArXivPetersen}.
\end{example}

As shown in the examples, \eqref{sol rank 1:single LMI region} is a necessary and sufficient condition to solve \eqref{probl single LMI region} for the open left halfplane and the open unit disk. 
On the other hand, Lemma~\ref{lemma:regions s=2 beta rank 1} improves on expressivity with respect to these two special cases by showing that we have equivalence with~\eqref{sol rank 1:single LMI region} for \emph{arbitrary} intersections of halfplanes left or right of \emph{any} real part, and disks with \emph{any} center on the real axis and \emph{any} radius.
When considering the regions expressed by $s=1$ and $s=2$, one can expect that increasing $s$ further yields more complex subsets of the complex plane, but finding their analytic expressions seems hardly tractable (numerical investigations aside), so the case $s>2$ is an open question.
Although Assumption~\ref{ass:rank 1} limits expressivity in the case $s=2$, we would like to show that by intersection of halfplanes and disks we can still obtain inner-approximations of subsets of the complex plane as in the next example.
Imposing eigenvalues within such an intersection for all matrices consistent with data is then without conservatism by virtue of Corollary~\ref{cor:nec and suff cond}.
\begin{example}
\label{example:wedge-approx}
Example~\ref{example:wedge} presented a relevant region for performance specifications and we show here a way to find an inner-approximation made only of intersections of disks.
Instead of examining all cases, we consider $\theta \in (0, \pi/4)$ and $\rho > \ell (3 + 2 \sqrt{2} ) \simeq 5.83 \ell$
since the first condition excludes insufficiently damped responses and the second one expresses that we do not need the magnitude $\rho$ of high-frequency closed-loop poles to be too close to the dominant closed-loop convergence rate $\ell$ (e.g., $\rho = 10 \ell$ would allow a decade between the two).
Abbreviate $\sin \theta$ and $\cos \theta$ with $\si_\theta$ and $\co_\theta$.
A circle centered on the negative real axis tangent to the cone $\{ \ri \co_\theta |y| = - \si_\theta |x|, x<0 \}$ can be parametrized by its center $x_{\tu{t}}$ and has equation $\{ \ri (x - x_{\tu{t}})^2 + y^2 = \si_\theta^2 x_{\tu{t}}^2 \}$.
The blue, green and red circles in Fig.~\ref{fig:circleMaxArea} are examples of such tangent circles.
Under the previous conditions on $\ell$, $\rho$, $\theta$, the green circle internally tangent to $\{ \ri x^2 + y^2 = \rho^2 \}$ has right end less than $-\ell$, hence the green disk has the largest intersection with $\mathcal{S}(\ell,\rho,\theta)$ among all other disks between it and the red disk.
As a second extreme, the blue disk has largest intersection with $\mathcal{S}(\ell,\rho,\theta)$ among all other disks with center farther to the left.
Then, the disk with the largest intersection with $\mathcal{S}(\ell,\rho,\theta)$ is between these two extremes (blue and green), and we need to consider surfaces like the cyan one.
A cyan surface has a break point with real part $\frac{\rho^2 + x_{\tu{t}}^2 \co_\theta^2}{2 x_{\tu{t}}} \in [-\rho,x_{\tu{t}}]$ and its \emph{half} area
\begin{align*}
& \frac{\pi \rho^2}{2} -\frac{1}{4}
\sqrt{(x_{\tu{t}}^2 (1+\si_\theta)^2 - \rho^2) (\rho^2-(1-\si_\theta)^2 x_{\tu{t}}^2)} \\
& - \frac{\rho^2}{2} \arccos \frac{x_{\tu{t}}^2 \co_\theta^2+\rho^2}{2 \rho x_{\tu{t}}} 
+\frac{\si_\theta^2 x_{\tu{t}}^2}{2} \arccos \frac{x_{\tu{t}}^2 (1-\si_\theta^2) - \rho^2}{2 \si_\theta x_{\tu{t}}^2}
\end{align*}
can be computed with an integral.
To find the best inner-approximation, we maximize this area with respect to $x_{\tu{t}}$, which has lower and upper bounds $-\rho/\co_\theta$ and $-\rho/(1+\si_\theta)$ (corresponding to the real part of centers of blue and green circles).
\end{example}

\begin{figure}
\centerline{\includegraphics[scale=.7]{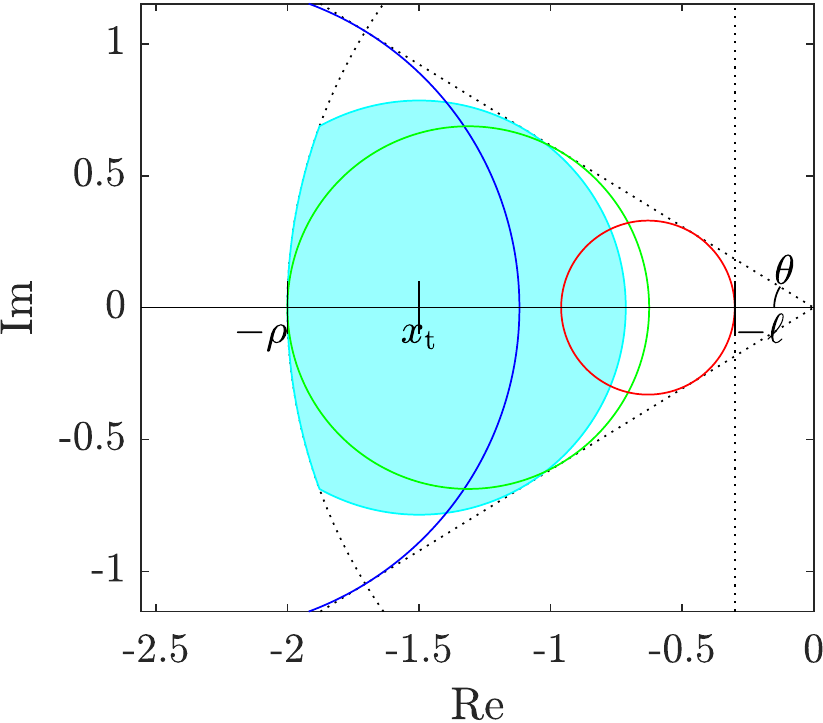}}
\caption{Finding largest intersection of two disks inner-approximating $\mathcal{S}(\ell,\rho,\theta)$ as explained in Example~\ref{example:wedge-approx}, for $\ell=0.3$, $\rho=2$, $\theta=\pi/5.7$.}
\label{fig:circleMaxArea}
\end{figure}

\begin{example}
\label{example:perf dt}
The discrete-time analogue of the continuous-time $\mathcal{S}(\ell,\rho,\theta)$ is in Fig.~\ref{fig:S(ell,rho,theta)}, right, and can not be represented by an LMI region since loci with constant natural frequency equal to $\rho$ and constant damping $\theta$ are obtained through complex exponentials \cite[\S 8.2.3]{franklin1994feedback}.
On the other hand, one can still give an inner-approximation through a disk as we will do in Section~\ref{sec:num ex:dt}.
\end{example}

\section{Alternative sufficient conditions}
\label{sec:suff cond alt}

In this section we briefly present an approach alternative to Theorem~\ref{thm:suff cond} to obtain sufficient conditions for $\mathcal{S}$-stability of all matrices consistent with data, and an approach that starts from a disturbance model based on an instantaneous bound rather than a bound on the energy of the disturbance sequence.
Both approaches rely on the S-procedure.

\subsection{Sufficient condition alternative to Theorem~\ref{thm:suff cond}}
\label{sec:suff cond S-proc}

Unlike Section~\ref{sec:suff cond}, we work directly with the first form of the set $\mathcal{C}$ of matrices consistent with data, see~\eqref{set C form 1}, and use the same strategy of replacing a block diagonal uncertainty with a full uncertainty.
We obtain then the next result.

\begin{proposition}
\label{prop:suff cond S-proc}
Let $\mathcal{S}$ be an LMI region with data $(\alpha,\beta)$. \eqref{probl single LMI region} is feasible if the next program is feasible
\begin{subequations}
\label{sol:suff cond S-proc}
\begin{align}
& \text{find} & & \hspace{-5pt} P=P^\top \succ 0, Y, \tau \ge 0 \label{sol:suff cond S-proc:positivity} \\
& \text{s.~t.} & & \hspace{-10pt}
\bmat{
\alpha \otimes P & \beta^\top \otimes \smat{P \\ Y}^\top \\
\beta \otimes \smat{P \\ Y} & 0}
- \tau 
\bmat{
I_s \otimes \mb{C} & I_s \otimes \mb{B}^\top \\
I_s \otimes \mb{B} & I_s \otimes \mb{A}
} \prec 0.
\label{sol:suff cond S-proc:lmi}
\end{align}
\end{subequations}
If \eqref{sol:suff cond S-proc} is feasible, the controller gain in~\eqref{probl single LMI region} is $K = Y P^{-1}$.
\end{proposition}
\begin{proof}
\eqref{sol:suff cond S-proc:lmi} implies by the S-procedure and \eqref{sol:suff cond S-proc:positivity} that
\begin{align*}
& 
\bmat{I_{sn} & \sa{Z}^\top}
\bmat{\alpha \otimes P & \beta^\top \otimes \smat{P\\ Y}^\top \\
\beta \otimes \smat{P\\ Y} & 0}
\bmat{I_{sn}\\ \sa{Z}} \prec 0\\
& \hspace*{4pt} \forall  \sa{Z} \colon \bmat{I_s \otimes I_n & \sa{Z}^\top} 
\bmat{I_s \otimes \mb{C} & I_s \otimes \mb{B}^\top\\ I_s \otimes \mb{B} & I_s \otimes \mb{A}}
\bmat{I_s \otimes I_n \\ \sa{Z}} \preceq 0.
\end{align*}
Since $\sa{Z}$ is a full uncertainty, the last condition implies that
\begin{align}
& \bmat{I_{sn} & \sa{Z}^\top}
\bmat{\alpha \otimes P & \beta^\top \otimes \smat{P\\ Y}^\top \\
\beta \otimes \smat{P\\ Y} & 0}
\bmat{I_{sn}\\ \sa{Z}} \prec 0   \label{suff S-proc:from block to full} \\
& \hspace*{3pt} \forall  \sa{Z}= I_s \otimes Z \colon  \notag \\
& \hspace*{6pt} \bmat{I_s \otimes I_n & I_s \otimes Z^\top} 
\bmat{I_s \otimes \mb{C} & I_s \otimes \mb{B}^\top\\ I_s \otimes \mb{B} & I_s \otimes \mb{A}}
\bmat{I_s \otimes I_n \\ I_s \otimes Z} \preceq 0 . \notag
\end{align}
The condition
\begin{align*}
& 0 \succeq \bmat{I_s \otimes I_n & I_s \otimes Z^\top} 
\bmat{I_s \otimes \mb{C} & I_s \otimes \mb{B}^\top\\ I_s \otimes \mb{B} & I_s \otimes \mb{A}}
\bmat{I_s \otimes I_n \\ I_s \otimes Z} \\
& = I_s \otimes \mb{C}  + (I_s \otimes Z^\top)(I_s \otimes \mb{B}) + (I_s \otimes \mb{B}^\top) (I_s \otimes Z) \\
& \hspace*{30pt} + (I_s \otimes Z^\top) (I_s \otimes \mb{A})( I_s \otimes Z ) \\
& = I_s \otimes (\mb{C} + Z^\top \mb{B} + \mb{B}^\top Z + Z^\top \mb{A} Z ) 
\end{align*}
is equivalent to
\begin{equation*}
0 \succeq \mb{C} + Z^\top \mb{B} + \mb{B}^\top Z + Z^\top \mb{A} Z =
\smat{I_n & Z^\top} \smat{\mb{C} & \mb{B}^\top \\ \mb{B} & \mb{A} }  \smat{I_n \\ Z}.
\end{equation*}
Hence, \eqref{suff S-proc:from block to full} is equivalent to
\begin{align*}
& \alpha \otimes P + \Tr \Big\{ \sa{Z}^\top \big( \beta \otimes \smat{P\\ Y} \big) \Big\} \prec 0 \\
& \hspace*{40pt} \forall  \sa{Z} = I_s \otimes Z \colon \smat{I_n & Z^\top} \smat{\mb{C} & \mb{B}^\top \\ \mb{B} & \mb{A} }  \smat{I_n \\ Z} \preceq 0 \\
& \!\!\iff \alpha \otimes P + \Tr \Big\{ \big( I_s \otimes Z^\top \big) \big( \beta \otimes \smat{P\\ Y} \big) \Big\} \prec 0 \\
& \hspace*{90pt} \forall Z \colon \smat{I_n & Z^\top} \smat{\mb{C} & \mb{B}^\top \\ \mb{B} & \mb{A} }  \smat{I_n \\ Z} \preceq 0 \\
& \!\!\iff 0 \succ \alpha \otimes P + \Tr \big\{ \big( I_s \otimes \smat{A & B} \big) \big( \beta \otimes \smat{P\\ Y} \big) \big\} \\
& \hspace*{31pt} = \alpha \otimes P + \Tr\big\{ \beta \otimes \big( \smat{A & B} \smat{P\\ Y} \big) \big\} \, \forall (A,B) \in \mathcal{C}
\end{align*}
by the form of $\mathcal{C}$ in~\eqref{set C form 1}.
\end{proof}

In the same way as Corollary~\ref{cor:suff cond} from Theorem~\ref{thm:suff cond}, one can obtain conditions for intersections of LMI regions from Proposition~\ref{prop:suff cond S-proc} easily, and it is thus omitted.

\subsection{Sufficient conditions for alternative disturbance model}
\label{sec:suff cond S-proc instant}

Unlike Section~\ref{sec:setting}, we consider a disturbance model involving not the energy of the disturbance sequence as~\eqref{set D}, but an instantaneous bound on the disturbance as
\begin{equation}
\label{set Di}
\mathcal{D}_{\tu{i}} := \{ d \in \real^n \colon dd^\top \preceq \epsilon I \},
\end{equation}
which expresses the bound $|d|^2 \le \epsilon$.
We refer the reader to \cite{bisoffi2021tradeoffs} for a discussion on the implications of instantaneous and energy bounds on disturbance for data-driven control.
The set of matrices consistent with a \emph{single} data point $i = 0, \dots, T-1$ is now
\begingroup
\thickmuskip=1.5mu plus 3mu
\begin{subequations}
\label{set Ci}
\begin{align}
& \mathcal{C}_i := \{ (A,B) \colon x^\circ (t_i) = A x(t_i) + B u(t_i) + d, d\in \mathcal{D}_{\tu{i}}  \} \\
& = \{ \bmat{A & B} = Z^\top \colon \bmat{I_n & Z^\top} \bmat{\mb{c}_i & \mb{b}_i^\top\\ \mb{b}_i & \mb{a}_i } \bmat{I_n\\ Z} \preceq 0 \} \label{set Ci:set only}\\
& \mb{c}_i : =  -\epsilon I + x^\circ(t_i) {x^\circ(t_i)}^\top \label{set Ci:c_i}\\
& \mb{b}_i : =  - \bmat{x(t_i)\\ u(t_i)} {x^\circ(t_i)}^\top , \mb{a}_i : = \bmat{x(t_i)\\ u(t_i)} \bmat{x(t_i) \\ u(t_i) }^\top, \label{set Ci:b_i and a_i}
\end{align}%
\end{subequations}%
\endgroup
cf.~\eqref{set C} and \eqref{set C form 1}.
Hence, the set of matrices consistent with \emph{all} data points is
\begin{equation}
\label{set I}
\mathcal{I} := \bigcap_{i \in \mathbb{I}} \mathcal{C}_i, \quad
\mathbb{I} := \{ 0,1, \dots, T-1 \}.
\end{equation}

\begin{remark}
We parallel here the observations in Remark~\ref{remark:different dist bound}.
Instead of~\eqref{set Di}, one can consider for the disturbance $d$ the instantaneous bound given by
\begin{equation*}
\mathcal{D}_{\tu{i}} := \{ d \in \real^n \colon \bmat{I & d} \bmat{r & s^\top \\ s & q} \bmat{I \\ d^\top} \preceq 0 \}
\end{equation*}
with matrix $r$ symmetric and scalar $q > 0$.
Moreover, one can obtain a set $\mathcal{C}_{\tu{i}}$ analogous to \eqref{set Ci}, namely, $\mathcal{C}_{\tu{i}}$ as in~\eqref{set Ci:set only} with, instead of \eqref{set Ci:c_i} and \eqref{set Ci:b_i and a_i},
\begin{align*}
& \hspace*{-8pt}\mb{c}_i : =  r + x^\circ(t_i) s + s^\top {x^\circ(t_i)}^\top + q \, x^\circ(t_i) {x^\circ(t_i)}^\top\\
& \hspace*{-8pt}\mb{b}_i : =  - \bmat{x(t_i)\\ u(t_i)} (s +  q {x^\circ(t_i)}^\top) , \mb{a}_i : = q \bmat{x(t_i)\\ u(t_i)} \bmat{x(t_i) \\ u(t_i) }^\top\!\!\!\!,
\end{align*}%
The rest of our results hold identically.
\end{remark}

Analogously to~\eqref{probl single LMI region}, we would like to solve
\begin{subequations}
\label{probl single LMI region intersec}
\begin{align}
& \text{find} & & \hspace*{-4pt} P = P^\top \succ 0, K \\
& \text{s.~t.} & &  \hspace*{-4pt} \alpha \otimes P + \beta \otimes \big( (A + B K ) P \big) \notag \\
& & & \hspace*{-4pt} \hspace*{2pt} + \beta^\top \otimes \big( P (A + B K)^\top \big) \prec 0 \quad \forall (A,B) \in \mathcal{I}.
\end{align}
\end{subequations}
With an approach similar to Section~\ref{sec:suff cond}, we can obtain the next result for the disturbance model with the instantaneous bound in~\eqref{set Di}.

\begin{proposition}
\label{prop:suff cond S-proc instant}
Let $\mathcal{S}$ be an LMI region with data $(\alpha,\beta)$.
\eqref{probl single LMI region intersec} is feasible if the next program is feasible
\begin{subequations}
\label{sol:suff cond S-proc instant}
\begingroup
\setlength{\arraycolsep}{1.5pt}
\thinmuskip=0.5mu plus 1mu
\medmuskip=1mu plus 2mu
\thickmuskip=1.5mu plus 3mu
\begin{align}
& \hspace{-5pt} \text{find} & & \hspace{-5pt} P=P^\top \succ 0, Y, \tau_0 \ge 0, \dots, \tau_{T-1} \ge 0 \label{sol:suff cond S-proc instant:positivity} \\
& \hspace{-5pt} \text{s.~t.} & & \hspace{-10pt}
\bmat{
\alpha \otimes P & \beta^\top \otimes \smat{P \\ Y}^\top \\
\beta \otimes \smat{P \\ Y} & 0}
- \sum_{i \in \mathbb{I} } \tau_i
\bmat{
I_s \otimes \mb{c}_i & I_s \otimes \mb{b}_i^\top \\
I_s \otimes \mb{b}_i & I_s \otimes \mb{a}_i
} \prec 0.
\label{sol:suff cond S-proc instant:lmi}
\end{align}
\endgroup
\end{subequations}
If \eqref{sol:suff cond S-proc instant} is feasible, the controller gain in~\eqref{probl single LMI region intersec} is $K = Y P^{-1}$.
\end{proposition}
\begin{proof}
\eqref{sol:suff cond S-proc instant:lmi} implies by the S-procedure and \eqref{sol:suff cond S-proc instant:positivity} that
\begin{align*}
& 
\bmat{I_{sn} & \sa{Z}^\top}
\bmat{
\alpha \otimes P & \beta^\top \otimes \smat{P \\ Y}^\top \\
\beta \otimes \smat{P \\ Y} & 0}
\bmat{I_{sn} \\ \sa{Z}} \prec 0 \\
& \hspace*{3pt} \forall  \sa{Z} \colon 
\bmat{I_s \otimes I_n & \sa{Z}^\top} 
\cdoT \bmat{I_s \otimes \mb{c}_i & I_s \otimes \mb{b}_i^\top\\ I_s \otimes \mb{b}_i & I_s \otimes \mb{a}_i}
[\star]^\top \preceq 0, i \in \mathbb{I},
\end{align*}
where for any matrices $G$ and $H = H^\top$, we use the notational shorthand $G \cdoT H [\star]^\top = G H G^\top$.
Since $\sa{Z}$ is a full uncertainty, the last condition implies
\begingroup
\setlength{\arraycolsep}{1.5pt}
\begin{align}
&
\bmat{I_{sn} & \sa{Z}^\top}
\bmat{
\alpha \otimes P & \beta^\top \otimes \smat{P \\ Y}^\top \\
\beta \otimes \smat{P \\ Y} & 0}
\bmat{I_{sn} \\ \sa{Z}} \prec 0 \label{suff cond S-proc inst to block} \\
& \hspace*{5pt} \forall  \sa{Z}= I_s \otimes Z \colon  \notag \\
& \hspace*{10pt} \bmat{I_s \otimes I_n & I_s \otimes Z^\top} 
\cdoT \bmat{I_s \otimes \mb{c}_i & I_s \otimes \mb{b}_i^\top\\ I_s \otimes \mb{b}_i & I_s \otimes \mb{a}_i} [\star]^\top \preceq 0,  i \in \mathbb{I}. \notag
\end{align}
\endgroup
Analogously to the proof of Proposition~\ref{prop:suff cond S-proc}, the condition
\begin{align*}
& 0 \succeq 
\bmat{I_s \otimes I_n & I_s \otimes Z^\top} \cdoT
\bmat{I_s \otimes \mb{c}_i & I_s \otimes \mb{b}_i^\top\\ I_s \otimes \mb{b}_i & I_s \otimes \mb{a}_i}
[\star]^\top, i \in \mathbb{I}
\end{align*}
is equivalent to
\begingroup
\setlength{\arraycolsep}{1.5pt}
\thinmuskip=0.7mu plus 1mu
\medmuskip=1.4mu plus 2mu
\thickmuskip=2.1mu plus 3mu
\begin{align*}
& 0 \succeq \mb{c}_i + Z^\top \mb{b}_i + \mb{b}_i^\top Z + Z^\top \mb{a}_i Z  = 
\smat{I_n & Z^\top}
\smat{\mb{c}_i & \mb{b}_i^\top \\ \mb{b}_i & \mb{a}_i }
\smat{I_n \\ Z}, ~i \in \mathbb{I};
\end{align*}
\endgroup
\eqref{suff cond S-proc inst to block} is thereby equivalent to
\begin{align*}
& \alpha \otimes P + \Tr \Big\{ \big( I_s \otimes Z^\top \big) \big( \beta \otimes \smat{P\\ Y} \big) \Big\} \prec 0 \\
& \hspace*{60pt} \forall Z \colon \smat{I_n & Z^\top} \smat{\mb{c}_i & \mb{b}_i^\top \\ \mb{b}_i & \mb{a}_i }  \smat{I_n \\ Z} \preceq 0, i \in \mathbb{I} \\
& \iff 
\alpha \otimes P + \Tr\big\{ \beta \otimes \big( \smat{A & B} \smat{P\\ Y} \big) \big\}  \prec 0 ~ \forall (A,B) \in \mathcal{I}
\end{align*}
by the definition of $\mathcal{I}$ in~\eqref{set I} and $\mathcal{C}_i$ in~\eqref{set Ci}.
\end{proof}

In the same way as Corollary~\ref{cor:suff cond} from Theorem~\ref{thm:suff cond}, one can obtain conditions for intersections of LMI regions from Proposition~\ref{prop:suff cond S-proc instant} easily, and it is thus omitted.

\section{Numerical investigation}
\label{sec:num ex}

For continuous and discrete time, we illustrate our findings and compare their feasibility.

\subsection{Continuous time}
\label{sec:num ex:ct}

\begin{figure}
\centerline{\includegraphics[scale=.6]{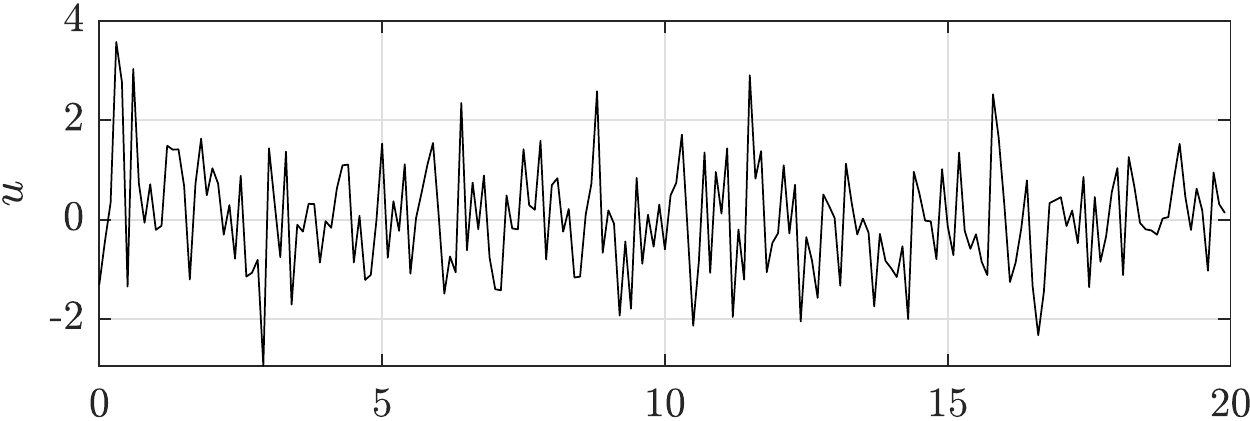}}
\centerline{\includegraphics[scale=.6]{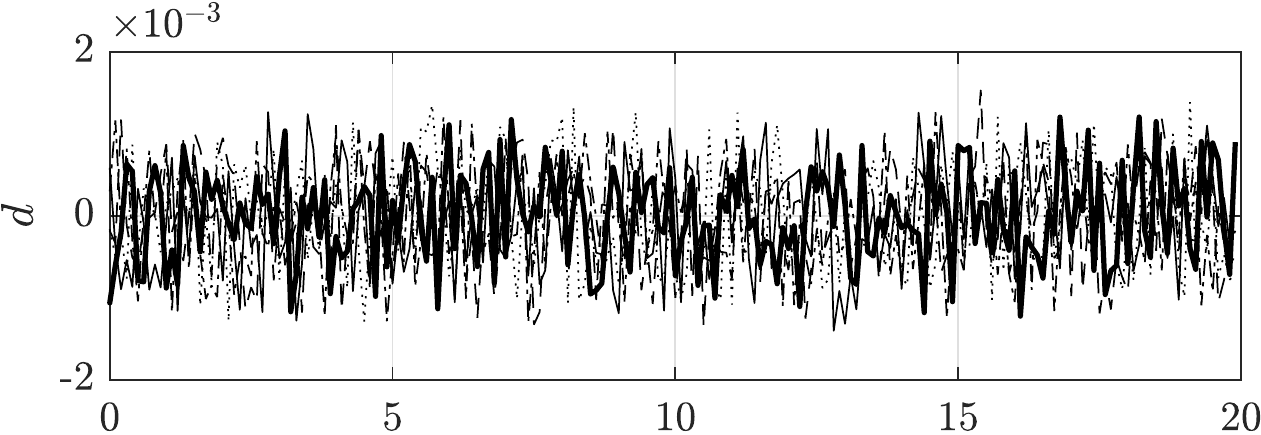}}
\centerline{\includegraphics[scale=.6]{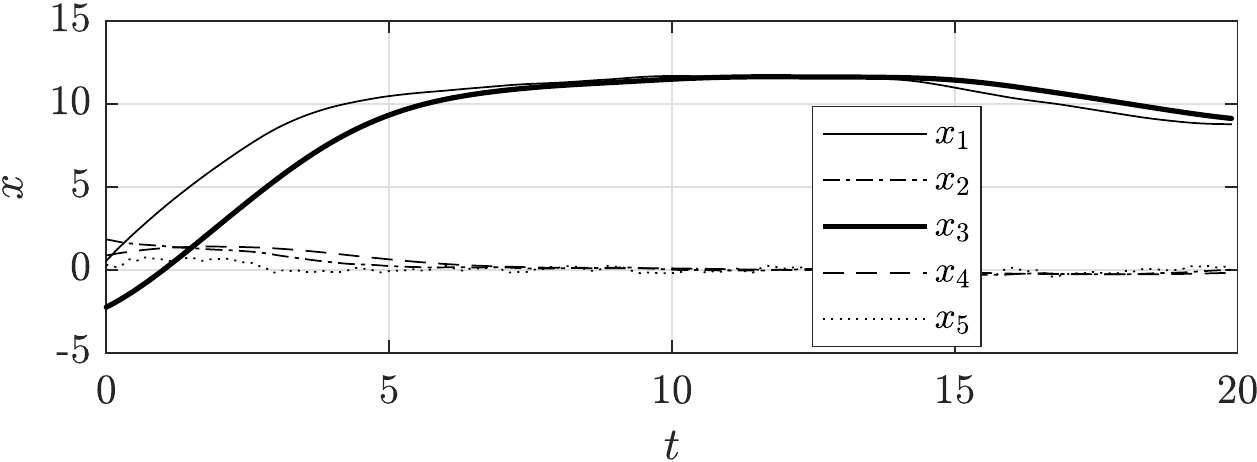}}
\caption{Experiment for data collection. The evolution of $d$ is not known and is reported only for completeness.}
\label{fig:exp}
\end{figure}

The following elements constitute our setting.
\begin{enumerate}[label=\textit{\arabic*)},left=0pt,wide]
\item We consider the \textit{dynamical system}
\begin{equation}
\label{sys:example}
\dot x =
\smat{
0      & 2      & 0     &  0      & 0\\
-0.1   & -0.35	& 0.1   &  0.1    & 0.75\\
0      & 0      & 0     &  2      & 0\\
0.4    & 0.4    & -0.4	&  -1.4   & 0\\
0      & -0.03	& 0     &  0      & -1
}
x +
\smat{0\\ 0\\ 0\\ 0\\ 1}
u +d
\end{equation}
taken from \cite[\S 9.5]{franklin1994feedback} (representing a digital tape transport). This model is used \emph{only} to generate data points and simulate the closed-loop response since no data-driven design relies on knowing it.
\item We know that the disturbance satisfies the squared norm bound $|d|^2 \le \epsilon$. This can be embedded in the \emph{disturbance model} $\mathcal{D}$ in~\eqref{set D} with $\Delta = \sqrt{T \epsilon} I$. 
Alternatively, it can be natively captured by the disturbance model $\mathcal{D}_{\tu{i}}$ in~\eqref{set Di}. 
We use $\epsilon=2.5 \cdot 10^{-6}$, so that $|d| \le \sqrt{\epsilon} = 1.58 \cdot 10^{-3}$.
\item The \emph{performance specification} is given by a region $\mathcal{S}(\ell, \rho, \theta)$, see Example~\ref{example:wedge}, with parameters $\ell=0.3$, $\rho=2$, $\theta=\pi/5.7$. 
As indicated in Example~\ref{example:wedge-approx}, this region can be inner-approximated by the intersection of the disks $\{ \ri x^2 + y^2 = \rho^2 \}$ and $\{ \ri (x - x_{\tu{t}})^2 + y^2 = \si_\theta^2 x_{\tu{t}}^2 \}$ with $x_{\tu{t}}= -1.4992$ (computed numerically).
\item A single \emph{experiment for data collection} is performed on~\eqref{sys:example}
under these conditions. 
The signals are in Fig.~\ref{fig:exp}. 
The input is obtained by interpolating linearly a sequence that is the realization of a Gaussian variable with mean zero and unit variance.
The disturbance is obtained in the same way from a realization of a random variable uniformly distributed in $|d| \le 1.58 \cdot 10^{-3}$ and is reported only for completeness since it is not accessible.
These continuous-time signals are then sampled with $T_{\tu{s}}=0.1$ to obtain $T=200$ data points of the data matrices in~\eqref{data matrices}.
\end{enumerate}

In the setting of the previous points, we compare on the same data set the designs proposed in the previous sections: 
\begin{itemize}[wide]
\item sufficient conditions in Corollary~\ref{cor:suff cond} and Proposition~\ref{prop:suff cond S-proc} for performance specification $\mathcal{S}(\ell, \rho, \theta)$ and disturbance model $\mathcal{D}$ in~\eqref{set D},
\item necessary and sufficient condition in Corollary~\ref{cor:nec and suff cond} for the inner-approximation of $\mathcal{S}(\ell, \rho, \theta)$ and $\mathcal{D}$ in~\eqref{set D},
\item sufficient condition in Proposition~\ref{prop:suff cond S-proc instant} for $\mathcal{S}(\ell, \rho, \theta)$ and disturbance model $\mathcal{D}_{\tu{i}}$ in~\eqref{set Di},
\item model-based condition in \eqref{sol:model based multiple} for $\mathcal{S}(\ell, \rho, \theta)$.
\end{itemize}
Note that Propositions~\ref{prop:suff cond S-proc}-\ref{prop:suff cond S-proc instant} need to be easily extended for the intersection of the LMI regions composing $\mathcal{S}(\ell, \rho, \theta)$.
All controller designs are obtained by YALMIP \cite{lofberg2004yalmip} and MOSEK ApS in MATLAB\textsuperscript{\textregistered} R2019b.

The resulting controller designs in terms of eigenvalues are in Fig.~\ref{fig:eigLocations}. 
All methods manage to move the eigenvalues into the desired $\mathcal{S}(\ell,\rho,\theta)$ or its inner-approximation, and the eigenvalue locations imposed by the different methods appear comparable. 
To appreciate differences, we show in Fig.~\ref{fig:cl} the time response of \eqref{sys:example} with $d=0$ in closed loop with a controller designed model-based or data-based with Corollary~\ref{cor:suff cond} or Proposition~\ref{prop:suff cond S-proc instant}. 
All responses are consistent with the specification imposed by $\mathcal{S}(\ell,\rho,\theta)$: e.g., the exhibited convergence rates of around $15$ time units are consistent with $\ell=0.3$. 
The model-based solution, which does not need to robustly stabilize a set of consistent matrices, shows on the other hand a smaller overshoot.

\begin{figure}
\centerline{\includegraphics[scale=.65]{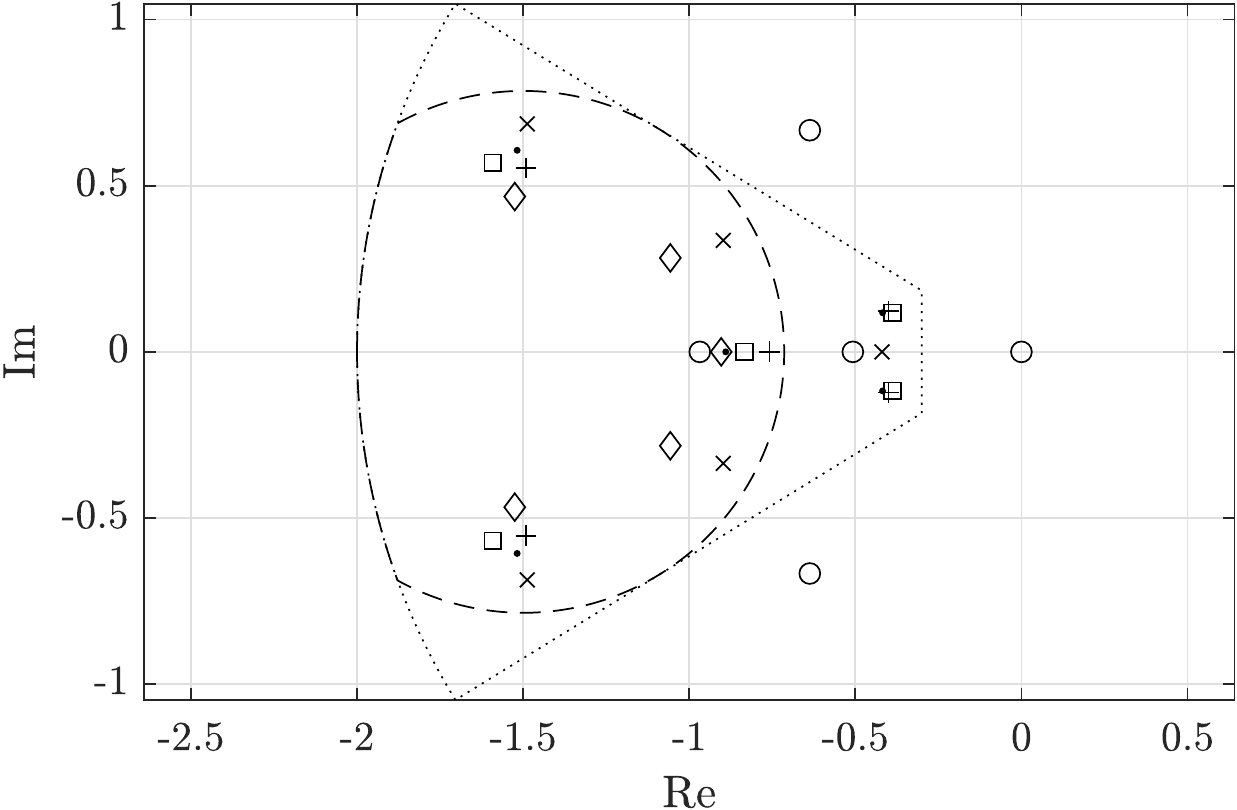}\hspace*{5mm}}
\caption{Eigenvalue locations: 
open loop ($\circ$), 
model-based ($\times$), 
data-based with Corollary~\ref{cor:suff cond} ($\square$),
with Corollary~\ref{cor:nec and suff cond} ($\diamond$),
with Proposition~\ref{prop:suff cond S-proc} ($+$), 
with Proposition~\ref{prop:suff cond S-proc instant} ($\boldsymbol{\cdot}$).
The region $\mathcal{S}(\ell,\rho,\theta)$ is dotted, and its inner-approximation dashed.}
\label{fig:eigLocations}
\end{figure}

\begin{figure}
\centerline{\includegraphics[scale=.6]{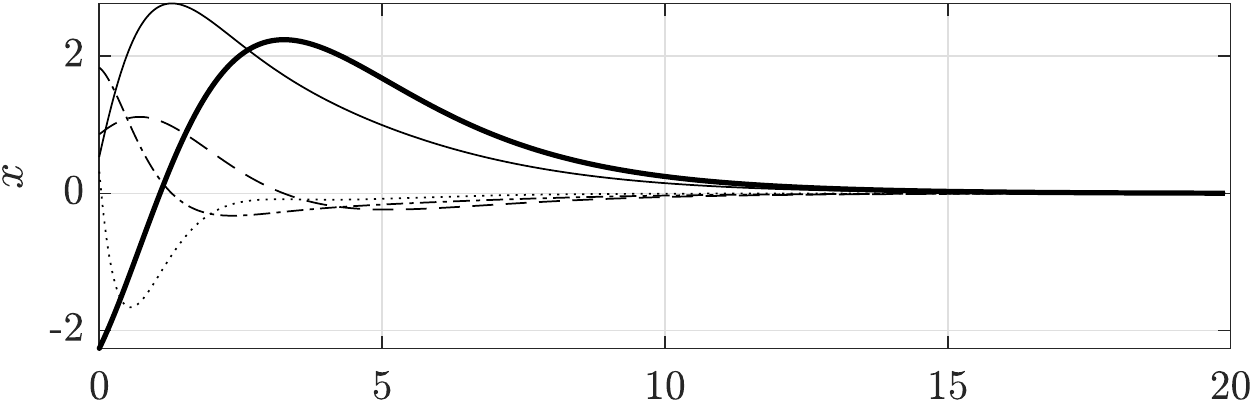}}
\centerline{\includegraphics[scale=.6]{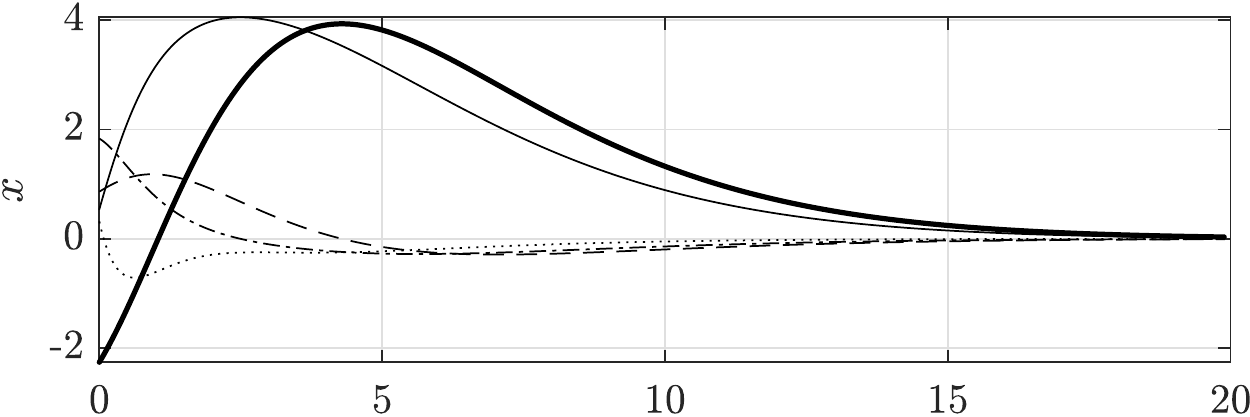}}
\centerline{\includegraphics[scale=.6]{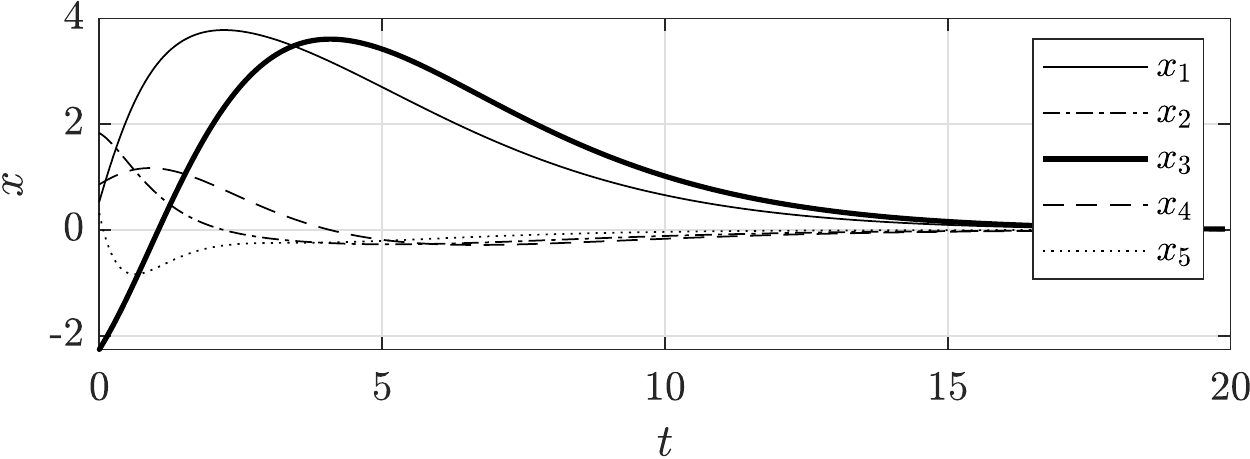}}
\caption{Closed-loop response of \eqref{sys:example} with controller: model-based (top), data-based and Corollary~\ref{cor:suff cond} (middle), data-based and Proposition~\ref{prop:suff cond S-proc instant} (bottom).}
\label{fig:cl}
\end{figure}

Next, we would like to examine which of the conditions can withstand the largest disturbance bound $\epsilon$ and summarize the results in Table~\ref{tab:epsilonVSmethods}.
For the specific system and performance specification, the method with largest robustness to the disturbance bound $\epsilon$ is given by Proposition~\ref{prop:suff cond S-proc instant}; this can be expected for the reasons discussed in~\cite{bisoffi2021tradeoffs} when the disturbance model is natively a norm bound $|d|^2\le \epsilon$. 
In particular, a disturbance model $\mathcal{D}$ obtained from this norm bound induces a set $\mathcal{C}$ that is at least as large as the set $\mathcal{I}$ induced by $\mathcal{D}_{\tu{i}}$ as in Section~\ref{sec:suff cond S-proc instant}. 
Moreover, the set $\mathcal{I}$ has the desirable property that it shrinks with an increasing $T$; $T=200$ is enough to make Proposition~\ref{prop:suff cond S-proc instant} more competitive than Proposition~\ref{prop:suff cond S-proc}, and a larger $T$ could withstand a larger $\epsilon$ at the expense of a higher computational cost due to the presence of $T$ extra variables $\tau_0$, \dots, $\tau_{T-1}$, see \cite{bisoffi2021tradeoffs}.
For methods working with $\mathcal{C}$, Proposition~\ref{prop:suff cond S-proc} and Corollary~\ref{cor:suff cond} are obtained in a conceptually similar way, but the former seems to have an edge over the latter.
Finally, although no conservatism is introduced by Corollary~\ref{cor:nec and suff cond} for the special LMI regions satisfying Assumption~\ref{ass:rank 1}, conservatism is introduced by inner-approximating the performance specification by special LMI regions, and this source of conservatism appears to be actually more significant.

\begin{table}
\centering
\begin{tabular}{rcccc}
\toprule
{$\epsilon$\hspace*{7mm}} & Prop.~\ref{prop:suff cond S-proc instant} & Prop.~\ref{prop:suff cond S-proc} & Cor.~\ref{cor:suff cond} & Cor.~\ref{cor:nec and suff cond} \\
\midrule
\grayRow $2.5\cdot 10^{-4}$	& \pseudoX		& \pseudoX		& \pseudoX		& \pseudoX\\
$1\cdot 10^{-4}$	& \checkmark	& \pseudoX		& \pseudoX		& \pseudoX\\
\grayRow $2.5\cdot 10^{-5}$	& \checkmark	& \checkmark	& \pseudoX		& \pseudoX\\
$1\cdot 10^{-5}$	& \checkmark	& \checkmark	& \checkmark	& \pseudoX\\
\grayRow $2.5\cdot 10^{-6}$	& \checkmark	& \checkmark	& \checkmark	& \checkmark\\
\bottomrule
\end{tabular}
\caption{For $\epsilon$, feasibility (\checkmark) or infeasibility (\pseudoX) of methods.}
\label{tab:epsilonVSmethods}
\end{table}

\subsection{Discrete time}
\label{sec:num ex:dt}

The following elements constitute our setting where we mention only those different than Section~\ref{sec:num ex:ct}.
\begin{enumerate}[label=\textit{\arabic*)},left=0pt,wide]
\item We consider the \textit{dynamical system}
\begin{equation}
\label{sys:example-DT}
x^+ = (I - \tfrac{1}{2} L) x + \smat{0\\ 0\\ 1\\ 0\\ 0} u + d,
L :=
\smat{
1	& 0		& -1	&  0	& 0\\
-1	& 1		& 0		&  0	& 0\\
0	& -1	& 1		&  0	& 0\\
0	& 0		& 0		&  1	& -1\\
-1	& 0		& 0     &  -1	& 2}
\end{equation}
where $L$ is the Laplacian matrix of an underlying digraph.

\item We consider the same \textit{disturbance models} as in Section~\ref{sec:num ex:ct}, now with $\epsilon=1 \cdot 10^{-5}$ and $|d| \le \sqrt{\epsilon}= 3.16 \cdot 10^{-3}$.

\item The \emph{performance specification} is given by a disk with center $(\frac{0.04+0.9}{2},0)$ and radius $\frac{-0.04+0.9}{2}$, which is a disk contained in the performance region described in Example~\ref{example:perf dt} and in Figure~\ref{fig:S(ell,rho,theta)}, right.

\item A single \emph{experiment for data collection} is performed on~\eqref{sys:example-DT} under these conditions. 
For $T=200$, the signals are in Fig.~\ref{fig:exp-DT}. 
The input is the realization of a Gaussian variable with mean zero and unit variance, the disturbance is the realization of a random variable uniformly distributed in $|d| \le 3.16 \cdot 10^{-3}$ and is reported only for completeness since it is not accessible.
\end{enumerate}

\begin{figure}
\centerline{\includegraphics[scale=.6]{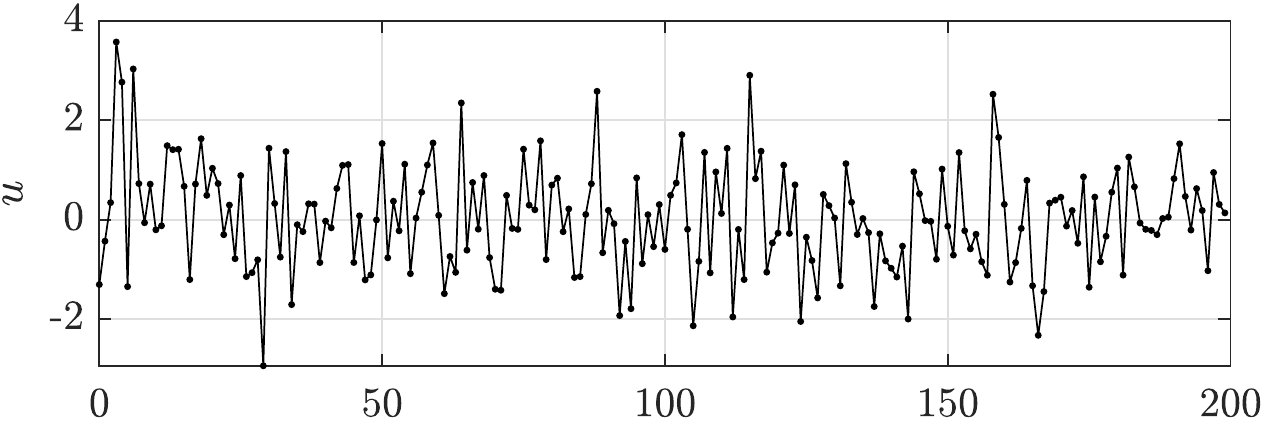}}
\centerline{\includegraphics[scale=.6]{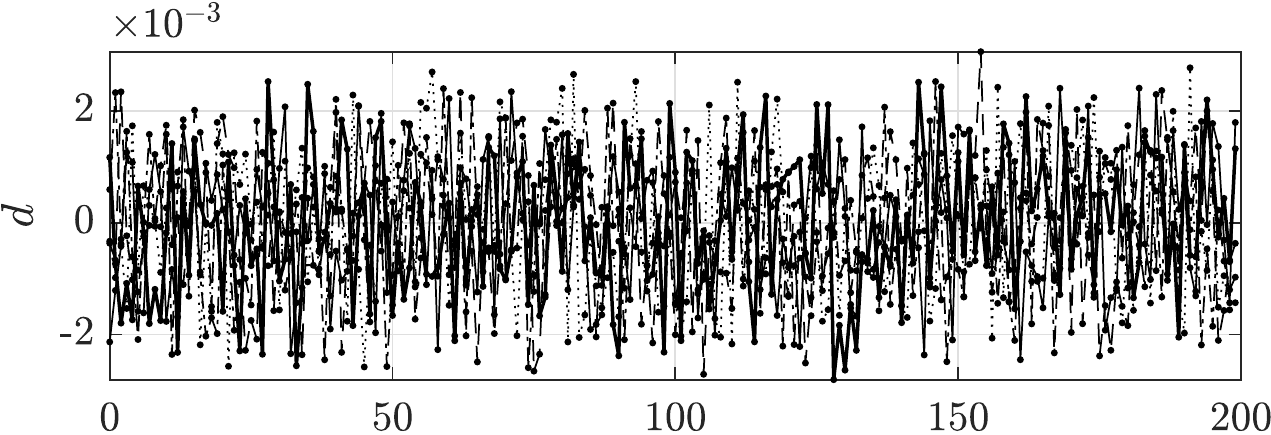}}
\centerline{\includegraphics[scale=.6]{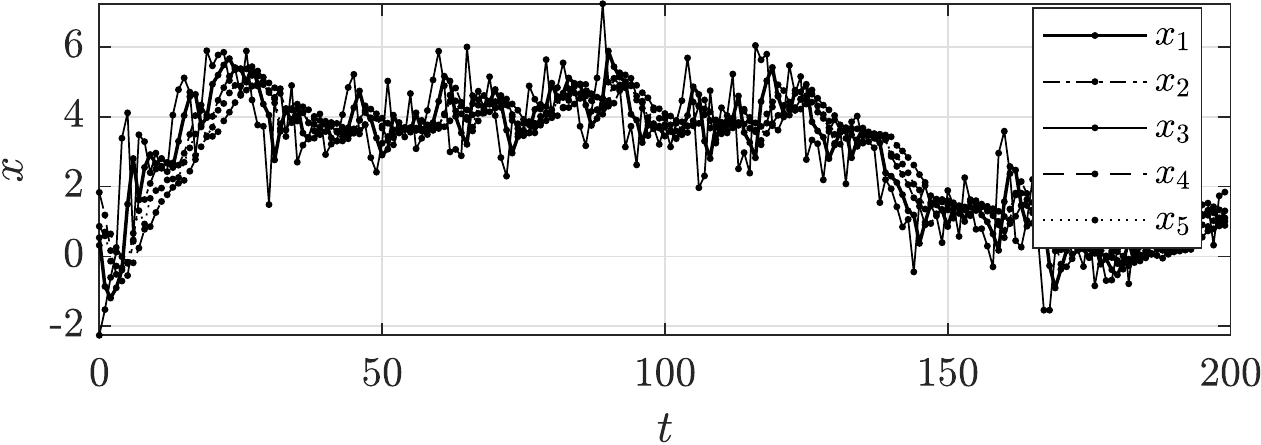}}
\caption{Experiment for data collection. The evolution of $d$ is not known and is reported only for completeness.}
\label{fig:exp-DT}
\end{figure}

Within this setting, we consider the same designs as in Section~\ref{sec:num ex:ct} and compare them for the disk giving the performance specification in discrete time.
The resulting controller designs in terms of eigenvalues are in Fig.~\ref{fig:eigLocations-DT}. 
All methods manage to move the eigenvalues into the desired disk, and the eigenvalue locations imposed by the different methods appear comparable. 
In Fig.~\ref{fig:cl-DT}, the time responses of \eqref{sys:example-DT} with $d=0$ in closed loop with a controller designed model-based or data-based with Corollary~\ref{cor:suff cond} or Proposition~\ref{prop:suff cond S-proc instant} are consistent with the specification and show a smaller overshoot in the model-based case.

\begin{figure}
\centerline{\includegraphics[scale=.65]{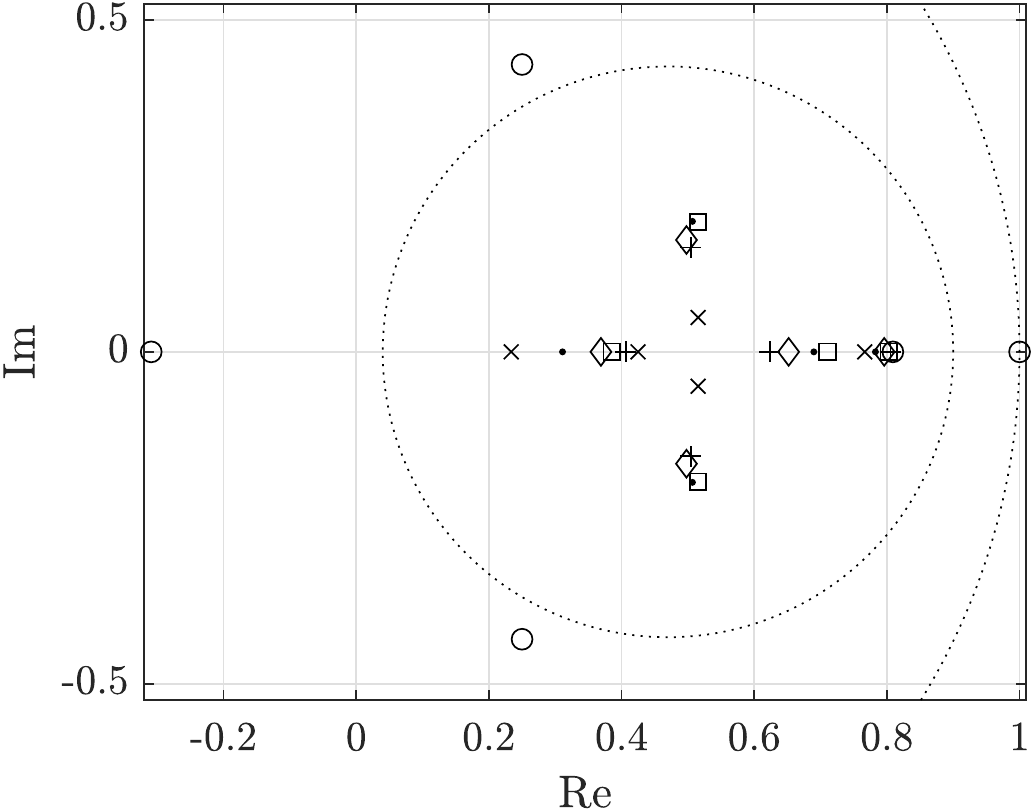}\hspace*{5mm}}
\caption{Eigenvalue locations: 
open loop ($\circ$), 
model-based ($\times$), 
data-based with Corollary~\ref{cor:suff cond} ($\square$),
with Corollary~\ref{cor:nec and suff cond} ($\diamond$),
with Proposition~\ref{prop:suff cond S-proc} ($+$), 
with Proposition~\ref{prop:suff cond S-proc instant} ($\boldsymbol{\cdot}$).
The unit and performance-specification disks are dotted.}
\label{fig:eigLocations-DT}
\end{figure}

\begin{figure}
\centerline{\includegraphics[scale=.6]{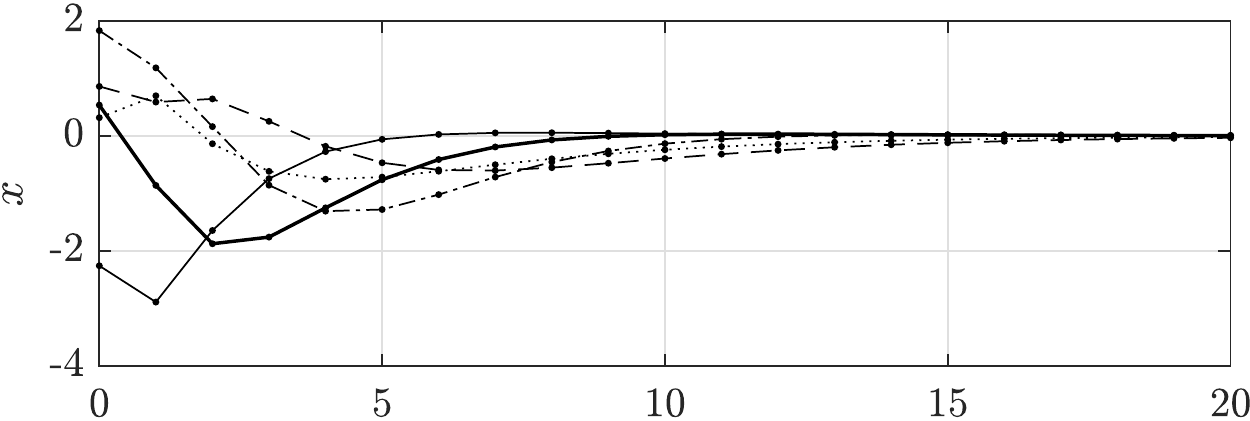}}
\centerline{\includegraphics[scale=.6]{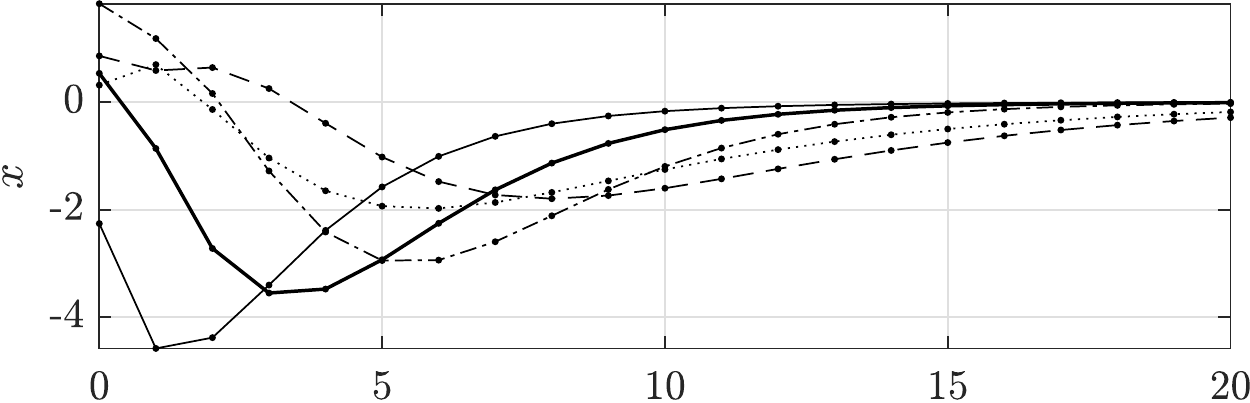}}
\centerline{\includegraphics[scale=.6]{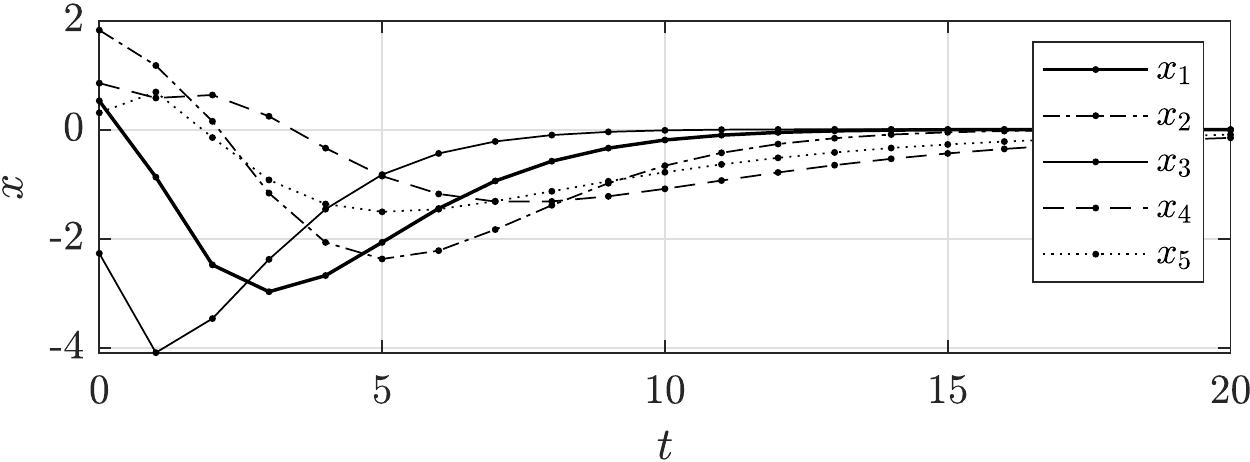}}
\caption{Closed-loop response of \eqref{sys:example-DT} with controller: model-based (top), data-based and Corollary~\ref{cor:suff cond} (middle), data-based and Proposition~\ref{prop:suff cond S-proc instant} (bottom).}
\label{fig:cl-DT}
\end{figure}

Next, we examine different values of $\epsilon$ as in Section~\ref{sec:num ex:ct}, and summarize the results in Table~\ref{tab:epsilonVSmethods-DT}.
The conclusions are analogous to those in Section~\ref{sec:num ex:ct}, except for the fact that now we no longer use Corollary~\ref{cor:nec and suff cond} on an inner-approximation. 
Indeed, when applied on a performance specification given by disk, Corollary~\ref{cor:nec and suff cond} gives a necessary and sufficient condition unlike the sufficient conditions of Corollary~\ref{cor:suff cond} and Proposition~\ref{prop:suff cond S-proc}, and is then able to withstand a larger $\epsilon$.

\begin{table}
\centering
\begin{tabular}{rcccc}
\toprule
{$\epsilon$\hspace*{7mm}} & Prop.~\ref{prop:suff cond S-proc instant} & Prop.~\ref{prop:suff cond S-proc} & Cor.~\ref{cor:suff cond} & Cor.~\ref{cor:nec and suff cond} \\
\midrule
\grayRow  $2.5\cdot 10^{-4}$	& \pseudoX		& \pseudoX		& \pseudoX		& \pseudoX\\
$1\cdot 10^{-4}$	& \checkmark	& \pseudoX		& \pseudoX		& \pseudoX\\
\grayRow $5\cdot 10^{-5}$	& \checkmark	& \pseudoX	& \pseudoX		& \checkmark\\
$2.5\cdot 10^{-5}$	& \checkmark	& \checkmark	& \pseudoX	& \checkmark\\
\grayRow $1\cdot 10^{-5}$	& \checkmark	& \checkmark	& \checkmark	& \checkmark\\
\bottomrule
\end{tabular}
\caption{For $\epsilon$, feasibility (\checkmark) or infeasibility (\pseudoX) of methods.}
\label{tab:epsilonVSmethods-DT}
\end{table}

\section*{Appendix}

The analytic expressions of the LMI regions with $s=2$ in Fig.~\ref{fig:basic LMI regions} are listed in Table~\ref{tab:sets}: for convenience, the last column is $2 \cdot \beta$ and, for an angle $\theta$, $\sin \theta$ or $\cos \theta$ are abbreviated as $\si_\theta$ or $\co_\theta$. We emphasize that the vertical halfplanes of the two initial rows could be obtained also for $s=1$ with, respectively, $\alpha = - l$ and $\beta = 1/2$ and $\alpha = r$ and $\beta = -1/2$.

\begin{table*}
\begin{tabular}{llll}
\toprule
Description & $\mathcal{S} = \{ \cnri$ & $\alpha$ & $2 \beta$\\
\midrule
\grayRow Vertical halfplane left of $l$ & $\quad x < l\}$ & $\smat{-l & 0\\0 & -1}$ & $\smat{1 & 0\\0 & 0}$\\
Vertical halfplane right of $r$ & $\quad x > r \}$ & $\smat{r & 0\\0 & -1}$ & $\smat{-1 & 0\\0 & 0}$\\
\grayRow Open disk with center $(x_{\tu{d}},0)$, radius $r_{\tu{d}}>0$ & $ \quad (x-x_{\tu{d}})^2 + y^2 < r_{\tu{d}}^2 \}$ & $\smat{-r_{\tu{d}} & x_{\tu{d}} \\  x_{\tu{d}}  & -r_{\tu{d}}}$ & $\smat{0 & 0\\ -2 & 0}$\\
Vertical strip with extremes $l<r$ & $ \quad l < x < r \}$ & $\smat{-r & 0\\0 & l}$ & $\smat{1 & 0\\0 & -1}$\\
\grayRow Horizontal strip with semiwidth $w>0$ & $\quad y^2 < w^2 \}$ & $\smat{-w & 0\\0 & -w}$ & $\smat{0 & 1\\-1 & 0}$\\
Ellipsoid with center $(x_{\tu{e}},0)$, semiaxes $\mu_1>0$ and $\mu_2>0$ & $\quad \frac{(x-x_{\tu{e}})^2}{\mu_1^2} + \frac{y^2}{\mu_2^2} < 1\}$ & $\smat{-\mu_1^2 & x_{\tu{e}} \mu_2 \\x_{\tu{e}} \mu_2 & -\mu_2^2}$ & $ \smat{0 & \mu_1-\mu_2 \\ -\mu_1 -\mu_2 & 0}$\\
\grayRow Left parabola with vertex $(x_{\tu{p}},0)$, curvature $c_{\tu{p}}>0$ & $\quad x< x_{\tu{p} } - (c_{\tu{p}}/2) y^2 \}$ & $\smat{-1 & 0\\0 & -x_{\tu{p}}}$ & $\smat{0 & \sqrt{\frac{c_{\tu{p}}}{2}}\\-\sqrt{\frac{c_{\tu{p}}}{2}} & 1}$\\
Right parabola with vertex $(x_{\tu{p}},0)$, curvature $c_{\tu{p}}>0$ & $\quad x> x_{\tu{p} } + (c_{\tu{p}}/2) y^2 \}$ & $\smat{-1 & 0\\0 & x_{\tu{p}}}$ & $\smat{0 & \sqrt{\frac{c_{\tu{p}}}{2}} \\ -\sqrt{\frac{c_{\tu{p}}}{2}} & -1}$\\
\grayRow Left hyperbola with vertex $(-x_{\tu{h}},0)$, asymptotes $\pm c_{\tu{h}} x$, $x_{\tu{h}}>0$, $c_{\tu{h}} > 0$& $\quad y^2 < c_{\tu{h}}^2 (x^2 - x^2_{\tu{h}}),x<0 \}$ & $\smat{0 & c_{\tu{h}} x_{\tu{h}} \\ c_{\tu{h}} x_{\tu{h}} & 0}$ & $\smat{c_{\tu{h}} & 1\\-1 & c_{\tu{h}}}$\\
Right hyperbola with vertex $(x_{\tu{h}},0)$, asymptotes $\pm c_{\tu{h}} x$, $x_{\tu{h}}>0$, $c_{\tu{h}} > 0$ & $\quad y^2 < c_{\tu{h}}^2 (x^2 - x^2_{\tu{h}}),x>0 \}$ & $\smat{0 & c_{\tu{h}} x_{\tu{h}} \\ c_{\tu{h}} x_{\tu{h}} & 0}$ & $\smat{-c_{\tu{h}} & 1\\-1 & -c_{\tu{h}}}$\\
\grayRow Left cone with vertex $(x_{\tu{c}},0)$, semiaperture $\theta \in (0,\pi/2)$ & $\quad \co_\theta |y| < \si_\theta (x_{\tu{c}} - x),x<x_{\tu{c}} \}$ & $ -\si_\theta x_{\tu{c}} \smat{1 & 0 \\ 0 & 1}$ & $\smat{\si_\theta & \co_\theta\\-\co_\theta & \si_\theta}$\\
Right cone with vertex $(x_{\tu{c}},0)$, semiaperture $\theta \in (0,\pi/2)$ & $\quad \co_\theta |y| < \si_\theta (x - x_{\tu{c}}),x > x_{\tu{c}} \}$ & $ \si_\theta x_{\tu{c}} \smat{1 & 0 \\ 0 & 1}$ & $ -\smat{\si_\theta & -\co_\theta\\ \co_\theta & \si_\theta}$\\
\bottomrule
\end{tabular}
\caption{Expressions of the LMI regions with $s=2$ in Fig.~\ref{fig:basic LMI regions}.}
\label{tab:sets}
\end{table*}

\bibliographystyle{plain}
\bibliography{pubs.bib}

\begin{thebibliography}{10}

\bibitem{berberich2019robust}
J.~Berberich, A.~Romer, C.~W. Scherer, and F.~Allg{\"o}wer.
\newblock Robust data-driven state-feedback design.
\newblock In {\em Proc. Amer. Control Conf.}, 2020.

\bibitem{berberich2020combining}
J.~Berberich, C.~W. Scherer, and F.~Allg{\"o}wer.
\newblock Combining prior knowledge and data for robust controller design.
\newblock {\em arXiv preprint arXiv:2009.05253}, 2020.

\bibitem{bisoffiArXivPetersen}
A.~Bisoffi, C.~{De Persis}, and P.~Tesi.
\newblock Data-driven control via {P}etersen’s lemma.
\newblock {\em arXiv preprint arXiv:2109.12175}, 2021.

\bibitem{bisoffi2021tradeoffs}
A.~Bisoffi, C.~{De Persis}, and P.~Tesi.
\newblock Trade-offs in learning controllers from noisy data.
\newblock {\em Systems \& Control Letters}, 2021.

\bibitem{chilali1996hInf}
M.~Chilali and P.~Gahinet.
\newblock {$H_\infty$} design with pole placement constraints: an {LMI}
  approach.
\newblock {\em IEEE Trans. Autom. Control}, 41(3):358--367, 1996.

\bibitem{chilali1999robust}
M.~Chilali, P.~Gahinet, and P.~Apkarian.
\newblock Robust pole placement in {LMI} regions.
\newblock {\em IEEE Trans. Autom. Control}, 44(12):2257--2270, 1999.

\bibitem{cocetti2020hybrid}
M.~Cocetti, S.~Donnarumma, L.~De~Pascali, M.~Ragni, F.~Biral, F.~Panizzolo,
  P.~Rinaldi, A.~Sassaro, and L.~Zaccarian.
\newblock Hybrid nonovershooting set-point pressure regulation for a wet
  clutch.
\newblock {\em IEEE/ASME Trans. Mechatr.}, 25(3):1276--1287, 2020.

\bibitem{coulson2018deepc}
J.~Coulson, J.~Lygeros, and F.~D{\"o}rfler.
\newblock Data-enabled predictive control: In the shallows of the {D}ee{PC}.
\newblock In {\em Proc. Eur. Control Conf.}, 2019.

\bibitem{dai2020ifac}
T.~Dai, M.~Sznaier, and B.~Roig~Solvas.
\newblock Data-driven quadratic stabilization of continuous {LTI} systems.
\newblock In {\em Proc. IFAC World Congress}, 2020.

\bibitem{depersis2020tac}
C.~{De Persis} and P.~Tesi.
\newblock Formulas for data-driven control: Stabilization, optimality and
  robustness.
\newblock {\em IEEE Trans. Autom. Control}, 65(3):909--924, 2020.

\bibitem{depersis2021lowcomplexity}
C.~{De Persis} and P.~Tesi.
\newblock Low-complexity learning of linear quadratic regulators from noisy
  data.
\newblock {\em Automatica}, 128, 2021.

\bibitem{franklin1994feedback}
G.~F. Franklin, J.~D. Powell, and A.~Emami-Naeini.
\newblock {\em Feedback control of dynamic systems - 3rd edition}.
\newblock Addison-Wesley, 1994.

\bibitem{horn1994topics}
R.~A. Horn and C.~R. Johnson.
\newblock {\em Topics in matrix analysis}.
\newblock Cambridge University Press, 1991.

\bibitem{lofberg2004yalmip}
J.~L{\"{o}}fberg.
\newblock {YALMIP}: A toolbox for modeling and optimization in {MATLAB}.
\newblock In {\em Proc. IEEE Int. Symp. Comp. Aid. Contr. Sys. Design}, 2004.

\bibitem{mania2019certainty}
H.~Mania, S.~Tu, and B.~Recht.
\newblock Certainty equivalence is efficient for linear quadratic control.
\newblock {\em arXiv preprint arXiv:1902.07826}, 2019.

\bibitem{olalla2011optimal}
C.~Olalla, I.~Queinnec, R.~Leyva, and A.~El~Aroudi.
\newblock Optimal state-feedback control of bilinear {DC}--{DC} converters with
  guaranteed regions of stability.
\newblock {\em IEEE Trans. Ind. Electr.}, 59(10):3868--3880, 2011.

\bibitem{pereira2013multiple}
L.~F.~A. Pereira, J.~V. Flores, G.~Bonan, D.~F. Coutinho, and J.~M.~G.
  da~Silva.
\newblock Multiple resonant controllers for uninterruptible power supplies —
  {A} systematic robust control design approach.
\newblock {\em IEEE Trans. Ind. Electr.}, 61(3):1528--1538, 2013.

\bibitem{petersen1987stabilization}
I.~R. Petersen.
\newblock A stabilization algorithm for a class of uncertain linear systems.
\newblock {\em Systems \& Control Letters}, 8(4):351--357, 1987.

\bibitem{poussotvassal2016gust}
C.~Poussot-Vassal, F.~Demourant, A.~Lepage, and D.~Le~Bihan.
\newblock Gust load alleviation: Identification, control, and wind tunnel
  testing of a {2-D} aeroelastic airfoil.
\newblock {\em IEEE Trans. Control Sys. Tech.}, 25(5):1736--1749, 2016.

\bibitem{sznaier2021survey}
M.~Sznaier.
\newblock Control oriented learning in the era of big data.
\newblock {\em IEEE Control Systems Letters}, 5(6):1855--1867, 2021.

\bibitem{vanwaarde2020noisy}
H.~J. van Waarde, M.~K. Camlibel, and M.~Mesbahi.
\newblock From noisy data to feedback controllers: non-conservative design via
  a matrix {S}-lemma.
\newblock {\em IEEE Trans. Autom. Control}, 2020.
\newblock Early Access.

\bibitem{willems2005note}
J.~C. Willems, P.~Rapisarda, I.~Markovsky, and B.~De~Moor.
\newblock A note on persistency of excitation.
\newblock {\em Systems \& Control Letters}, 54(4):325--329, 2005.

\bibitem{xue2020datadriven}
A.~Xue and N.~Matni.
\newblock Data-driven system level synthesis.
\newblock In {\em Proc. 3rd Conf. Learning for Dynam. Contr.}, 2021.

\end{thebibliography}

\end{document}